\documentclass{llncs}

\usepackage{url}
\usepackage{graphicx}
\usepackage{amsmath,amssymb,xcolor,wrapfig,paralist}
\usepackage{algorithmicx,algorithm}
\usepackage[noend]{algpseudocode}
\usepackage{multirow}

\usepackage{booktabs}
\usepackage{proof}
\usepackage{tikz}
\usepackage{colortbl}
\usepackage{pdflscape}
\usetikzlibrary{shapes,positioning,arrows,calc,automata,matrix}
\tikzstyle{state}+=[minimum size = 8mm, inner sep=0,outer sep=1]
\tikzset{->,>=stealth'}

\definecolor{wwhite}{gray}{1}
\usepackage[T1]{fontenc}
\newcommand*{\ditto}{---\textquotedbl---}

\usepackage{tabularx}
\newcolumntype{L}{>{\raggedright\arraybackslash}p{1.6cm}}
\newcolumntype{C}{>{\centering\arraybackslash}p{1.6cm}}
\newcolumntype{R}[1]{>{\raggedleft\arraybackslash}p{#1}}
\allowdisplaybreaks

\newcommand{\thmhelperpre}[2]{\newcommand{\theoremlike}[1]{\par\medskip\penalty-250\refstepcounter{theorem}{\bfseries\noindent##1 \ref{#1}.}\itshape}\theoremlike{#2}}
\newcommand{\thmhelperpost}{\par\medskip%
 \renewcommand{\theoremlike}[1]{\par\medskip\penalty-250\refstepcounter{theorem}{\bfseries\noindent##1 \thesection .\thetheorem.}\itshape}%
}

\newcommand{\QED}{\qed}

\newcommand{\qee}{\hfill$\triangle$} 

\newcommand{\para}[1]{\medskip \noindent\textbf{#1.}}
\newcommand{\pr}{\mathbb P}
\newcommand{\expected}{\mathbb E}

\newcommand{\pmin}{p_{\mathsf{min}}}
\newcommand{\pt}{p_{\mathsf{term}}}
\newcommand{\Nset}{\mathbb N}

\newcommand{\Mc}{\mathcal{M}}
\newcommand{\Pm}{\mathbf{P}}
\newcommand{\St}{S}
\newcommand{\init}{\mu}
\newcommand{\Lab}{L}

\newcommand{\dra}{\mathcal{A}}
\newcommand{\draS}{Q}
\newcommand{\draAl}{{2^{Ap}}}
\newcommand{\draTr}{\gamma}
\newcommand{\draInit}{q_o}
\newcommand{\draAcc}{Acc}

\newcommand{\ccand}{C}

\newcommand{\bscc}{\mathsf{BSCC}}
\newcommand{\scs}{\mathsf{SC}}
\newcommand{\run}{\rho}
\renewcommand{\path}{\pi}
\newcommand{\emptypath}{\lambda}
\newcommand{\concat}{\,.\,}
\newcommand{\runs}{\mathsf{Runs}}
\newcommand{\candidatepath}{\kappa}
\newcommand{\candidate}{K}
\newcommand{\kcand}[1]{\mathit{Cand}_{#1}}
\newcommand{\skcand}[1]{\mathit{SCand}_{#1}}
\newcommand{\support}[1]{\overline{#1}}

\newcommand{\comp}{C}
\DeclareMathOperator{\Inf}{Inf}
\newcommand{\reach}{\Diamond}
\newcommand{\alws}{\Box}
\newcommand{\MP}{\mathsf{MP}}

\newcommand{\X}{{\ensuremath{\mathbf{X}}}}

\newcommand{\U}{{\ensuremath{\mathbf{U}}}}

\newcommand{\True}{\mathbf{Yes}}
\newcommand{\False}{\mathbf{No}}

\newcommand{\cone}{\mathsf{Cone}}
\newcommand{\nextState}{\mathsf{NextState}}

\newcommand{\trerr}[1]{\xi}
\newcommand{\mperr}[1]{\zeta}

\usepackage{microtype}
\newcommand{\mybigspace}{}

\title{Faster Statistical Model Checking for Unbounded Temporal Properties
\thanks{This research was funded in part by the European Research Council (ERC) under grant agreement 267989 (QUAREM), the Austrian Science Fund (FWF) under grants project S11402-N23 (RiSE) and Z211-N23 (Wittgenstein Award), the People Programme (Marie Curie Actions) of the European Union's Seventh Framework Programme (FP7/2007-2013) REA Grant No 291734, the SNSF Advanced Postdoc. Mobility Fellowship -- grant number P300P2\_161067, and the Czech Science Foundation under grant agreement P202/12/G061.}
}
\author{Przemys\l aw Daca\inst1 \and Thomas A.\ Henzinger\inst1 \and Jan K{{\v r}}et\'insk\'y\inst2 \and Tatjana Petrov\inst1}
\institute{IST Austria \and Institut f\"ur Informatik, Technische Universit\"at M\"unchen, Germany}

\begin{document}

\pagestyle{plain}
\maketitle

\mybigspace\mybigspace

\begin{abstract}
We present a new algorithm for the statistical model checking of Markov chains with respect to unbounded temporal properties, including full linear temporal logic. The main idea is that we monitor each simulation run on the fly, in order to detect quickly if a bottom strongly connected component is entered with high probability, in which case the simulation run can be terminated early. As a result, our simulation runs are often much shorter than required by termination bounds that are computed a priori for a desired level of confidence on a large state space. In comparison to previous algorithms for statistical model checking our method is not only faster in many cases but also requires less information about the system, namely, only the minimum transition probability that occurs in the Markov chain. 
In addition, our method can be generalised to unbounded quantitative properties such as mean-payoff bounds.
\end{abstract}

\mybigspace


\mybigspace
\section{Introduction}
\label{sec:intro}

Traditional numerical algorithms for the verification of Markov chains may be 
computationally intense or inapplicable, when facing a large state space or 
limited knowledge about the chain. 
To this end, statistical algorithms are used as a powerful alternative.
\emph{Statistical model checking} (SMC) typically refers to approaches where 
(i) 
finite paths of the Markov chain are sampled a finite number of times, 
(ii) 
the property of interest is verified for each sampled path 
(e.g.\ state $r$ is reached),
and 
(iii) 
hypothesis testing or statistical estimation is used to infer conclusions 
(e.g.\ state $r$ is reached with probability at most~$0.5$) 
and give statistical guarantees 
(e.g.\ the conclusion is valid with $99\%$ confidence).
SMC approaches differ in 
(a) 
the class of properties they can verify 
(e.g.\ bounded or unbounded properties),
(b)
the strength of statistical guarantees they provide 
(e.g.\ confidence bounds, only asymptotic convergence of the method towards the correct value, or none), and
(c) 
the amount of information they require about the Markov chain 
(e.g.\ the topology of the graph).
In this paper, we provide an algorithm for SMC of unbounded properties, with confidence bounds, 
in the setting where only the minimum transition probability of the chain is known. 
Such an algorithm is particularly desirable in scenarios when the system is not known (``black box''), 
but also when it is too large to construct or fit into memory.

Most of the previous efforts in SMC has focused on the analysis of properties 
with bounded horizon 
\cite{Younes02,Sen04,DBLP:journals/sttt/YounesKNP06,DBLP:conf/cmsb/JhaCLLPZ09,DBLP:conf/tacas/JegourelLS12,DBLP:journals/corr/abs-1207-1272}.
For bounded properties 
(e.g.\ state $r$ is reached with probability at most $0.5$ in the first 
$1000$ steps) 
statistical guarantees can be obtained in a completely black-box setting, 
where execution runs of the Markov chain can be observed, but no other 
information about the chain is available.
Unbounded properties 
(e.g.\ state $r$ is reached with probability at most $0.5$ in any number of 
steps) 
are significantly more difficult, as a stopping criterion is needed when 
generating a potentially infinite execution run, and some information about 
the Markov chain is necessary for providing statistical guarantees
(for an overview, see Table~\ref{tab:approaches}).
On the one hand, some approaches require the knowledge of the full topology in order to preprocess the Markov chain.
On the other hand, when the topology is not accessible, there are approaches where the correctness of the statistics relies on information ranging from
the second eigenvalue $\lambda$ of the Markov chain, to knowledge of both the number $|S|$ of states and the 
minimum transition probability~$\pmin$.  

\begin{table}[t]

\caption{SMC approaches to Markov chain verification, organised
by (i) the class of verifiable properties, and (ii) by the required information about the Markov chain, where
$\pmin$ is the minimum transition probability,
$|S|$ is the number of states, and 
$\lambda$ is the second largest eigenvalue of the chain.
}\label{tab:approaches}
\begin{tabular}{l|CCCCC}
 {LTL, mean payoff} 		& $\times$ 			& \textbf{here} 	& \mbox{\cite{atva14} (LTL)} 	& 			&\\
 {$\Diamond, \mathbf{U}$} 	& $\times$ 			& \textbf{here}   	& \ditto 			& \cite{sbmf11} 	& \cite{sbmf11}, \cite{ase10} \\
  {bounded} 			& {e.g.~\cite{Younes02,Sen04}} 	& & & &\\
\hline
& {no info} & $\pmin$ & $|S|,\pmin$  &   $\lambda$ & {topology} \\
\end{tabular}

\end{table}

\textbf{Our contribution} is a new SMC algorithm for full linear temporal logic (LTL), 
as well as for unbounded quantitative properties (mean payoff), which provides strong 
guarantees in the form of confidence bounds.
Our algorithm uses less information about the Markov chain than previous 
algorithms that provide confidence bounds for unbounded properties%
---we need to know only the minimum transition probability $\pmin$ of the 
chain, and not the number of states nor the topology.
Yet, experimentally, our algorithm performs in many cases better than these previous 
approaches (see Section~5).
Our main idea is to \emph{monitor each execution run on the fly in order to 
build statistical hypotheses about the structure of the Markov chain}.
In particular, if from observing the current prefix of an execution run we can 
stipulate that with high probability a bottom strongly connected component 
(BSCC) of the chain has been entered, then we can terminate the current 
execution run.  
The information obtained from execution prefixes allows us to terminate executions as soon as the property is decided with the required confidence, which is usually much earlier than any bounds 
that can be computed a priori. As far as we know, this is the first SMC algorithm that uses information 
obtained from execution prefixes.

Finding $\pmin$ is a light assumption in many realistic scenarios and often does not depend on the size of the chain -- 
e.g. bounds on the rates for reaction kinetics in chemical reaction systems are typically known, from a \textsc{Prism} language model they can be easily inferred without constructing the respective state space.

\begin{example}
\label{ex:motivating}
Consider the property of reaching state $r$ in the Markov chain depicted in 
Figure~\ref{fig:ex-MC-intro}. 
While the execution runs reaching $r$ satisfy the property and can be stopped 
without ever entering any $v_i$, the finite execution paths without $r$, such as 
$stuttutuut$, are inconclusive. 
In other words, observing this path does not rule out the existence of a 
transition from, e.g., $u$ to $r$, which, if existing, would eventually be 
taken with probability $1$. 
This transition could have arbitrarily low probability, rendering its 
detection arbitrarily unlikely, yet its presence would change the probability 
of satisfying the property from $0.5$ to~$1$.
However, knowing that if there exists such a transition leaving the set, its 
transition probability is at least $\pmin=0.01$, we can estimate the probability 
that the system is stuck in the set $\{t,u\}$ of states. 
Indeed, if existing, the exit transition was missed at least four times, no matter 
whether it exits $t$ or $u$.
Consequently, the probability that there is no such transition and $\{t,u\}$ 
is a BSCC is at least $1-(1-\pmin)^4$.
\end{example}

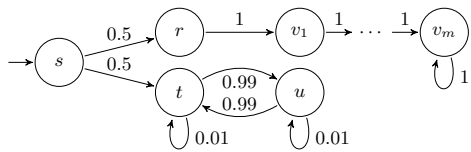
\begin{figure}
\centering
\scalebox{0.8}{
\begin{tikzpicture}
\node[state,initial,initial text=] (s) at (0,0.5){$s$};
\node[state] (r) at (2,1){$r$};
\node[state] (v1) at (4,1){$v_1$};
\node (vi) at (5.2,1){$\cdots$};
\node[state] (vn) at (6.4,1){$v_m$};
\path[->]
(v1) edge node[above]{$1$} (vi)
(vi) edge node[above]{$1$} (vn)
(vn) edge[loop below] node[right,pos=0.2]{$1$} ();
\node[state] (t) at (2,0){$t$};
\node[state] (u) at (4,0){$u$};
\path[->] 
(s) edge node[above]{$0.5$} (r)
(s) edge node[above]{$0.5$} (t)
(r) edge node[above]{$1$} (v1)
(t) edge[bend left] node[below]{$0.99$} (u)
(u) edge[bend left] node[above]{$0.99$} (t)
(t) edge[loop below] node[right,pos=0.2]{$0.01$} ()
(u) edge[loop below] node[right,pos=0.2]{$0.01$} ()
;
\end{tikzpicture}
}
\caption{A Markov chain.}\label{fig:ex-MC-intro}
\end{figure}

This means that, in order to get $99\%$ confidence that $\{t,u\}$ is a BSCC,
we only need to see both $t$ and $u$ around $500$ 
times\footnote{$1-(1-\pmin)^{500}=1-0.99^{500}\approx 0.993$} 
on a run.
This is in stark contrast to a priori bounds that provide the same level of 
confidence, such as the $(1/\pmin)^{|S|}=100^{\mathcal O(m)}$ runs 
required by \cite{atva14}, which is infeasible for large $m$.
In contrast, our method's performance is independent of $m$.
\qee \smallskip

Monitoring execution prefixes allows us to design an SMC algorithm  
for complex unbounded properties such as full LTL.
More precisely, we present a new SMC algorithm for LTL 
over Markov chains, specified as follows:
\begin{description}
\item [\textbf{Input:}] 
  we can sample finite runs of arbitrary length from an unknown finite-state 
  discrete-time Markov chain $\Mc$ starting in the initial state\footnote{%
    We have a black-box system in the sense of \cite{Sen04}, different from
    e.g.\ \cite{Younes02} or \cite{atva09}, where simulations can  be run from 
    any state.}, and
  we are given a lower bound $\pmin>0$ on the transition probabilities 
  in~$\Mc$, an LTL formula~$\varphi$, a threshold probability~$p$, an
  indifference region $\varepsilon>0$, and two error bounds 
  $\alpha,\beta>0$,\footnote{Except for the transition probability 
    bound~$\pmin$, all inputs are standard, as used in 
    literature, e.g.~\cite{Younes02}.}
\item [\textbf{Output:}] 
  if $\pr[\Mc\models \varphi]\geq p+\varepsilon$, return YES with probability 
  at least $1-\alpha$, and\\
  \hspace*{9mm}if $\pr[\Mc\models \varphi]\leq p-\varepsilon$, return NO with 
  probability at least $1-\beta$. 
\end{description}
In addition, we present the first SMC algorithm for computing the mean payoff of Markov chains whose states are labelled with rewards. 

\textbf{Related work.} 
To the best of our knowledge, we present the first SMC algorithm that provides confidence bounds for unbounded qualitative properties with access to only the minimum probability of the chain $\pmin$,
and the first SMC algorithm for quantitative properties.
For completeness, we survey briefly other related SMC approaches. 
SMC of unbounded properties, usually ``unbounded until'' properties, was 
first considered in \cite{vmcai04} and the first approach was proposed in \cite{cav05}, but observed incorrect in \cite{ase10}.
Notably, in~\cite{sbmf11} two approaches are described. 
The first approach proposes to terminate sampled paths at every step with some probability $\pt$.
In order to guarantee the asymptotic convergence of this method, the second eigenvalue $\lambda$ of 
the chain must be computed, which is as hard as the verification problem itself.
It should be noted that their method provides only asymptotic guarantees as the width of the confidence interval converges to zero.
The correctness of \cite{DBLP:journals/apal/LassaigneP08} 
relies on the knowledge of the second eigenvalue $\lambda$, too.
The second approach of~\cite{sbmf11} requires the knowledge of the chain's 
topology, which is used to transform the chain so that all potentially 
infinite paths are eliminated.
In \cite{ase10}, a similar transformation is performed, again requiring 
knowledge of the topology.
The (pre)processing of the state space required by the topology-aware methods,
 as well as by traditional numerical methods for Markov chain analysis, is a major practical hurdle for large (or unknown) state spaces.
In \cite{atva14} a priori bounds for the length of execution runs are 
calculated from the minimum transition probability and the number of states.
However, without taking execution information into account, these bounds are 
exponential in the number of states and highly impractical, as illustrated in 
the example above.
Another approach, limited to ergodic Markov chains, is taken in \cite{atva09}, based on coupling methods.
There are also extensions of SMC to timed systems \cite{DBLP:journals/sttt/DavidLLMP15}.
Our approach is also related to \cite{GrosuS05,OudinetDGLP11}, where the product of a non-deterministic system and B\"{u}chi automaton is explored for accepting lassos.
We are not aware of any method for detecting BSCCs by observing a single run, employing no directed search of the state space.


\textbf{Experimental evaluation.}  Our idea of inferring the structure
of the Markov chain on the fly, while generating execution runs,
allows for their early termination.  In Section~5 we will see that for
many chains arising in practice, such as the concurrent probabilistic
protocols from the \textsc{Prism} benchmark suite \cite{benchmarks},
the BSCCs are reached quickly and, even more importantly, can be small
even for very large systems.  Consequently, many execution runs can be
stopped quickly.  Moreover, since the number of execution runs
necessary for a required precision and confidence is independent of the size
of the state space, it needs not be very large even for highly confident
results (a good analogy is that of the opinion polls: the precision and confidence of opinion polls is regulated by the sample size and is independent of the size of the population). 
It is therefore not surprising
that, experimentally, in most cases from the benchmark suite, our
method outperforms previous methods (often even the numerical methods)
despite requiring much less knowledge of the Markov chain, and despite
providing strong guarantees in the form of confidence bounds.  In
Section~\ref{sec:discus}, we also provide theoretical bounds on the
running time of our algorithm for classes of Markov chains on which it
performs particularly well.

\mybigspace
\section{Preliminaries}
\label{sec:prelim}

\begin{definition} [Markov chain]
A \emph{Markov chain (MC)} is a tuple $\Mc = (\St, \Pm, \init)$, where
\begin{itemize}
\item $\St$ is a finite set of states,
\item $\Pm \;:\; \St \times \St \to [0,1]$ is the transition probability matrix, such that for every $s\in \St$ it holds $\sum_{s'\in \St} \Pm(s,s') = 1$,
\item $\init$ is a probability distribution over $\St$. 
\end{itemize}
\end{definition}
We let $\pmin:=\min(\{\Pm(s,s') > 0\mid s,s'\in\St\})$ denote the smallest positive transition probability in $\Mc$.
A \emph{run} of $\Mc$ is an infinite sequence $\run = s_0 s_1 \cdots$ of states, such that 
for all $i\geq 0$, $\Pm(s_i, s_{i+1}) > 0$; we let $\run[i]$ denote the state $s_i$. 
A \emph{path} in $\Mc$ is a finite prefix of a run of $\Mc$. We denote the empty path by $\emptypath$ and concatenation of paths $\path_1$ and $\path_2$ by $\path_1\concat\path_2$.
Each path $\path$ in $\Mc$ determines the set of runs $\cone(\path)$ consisting of all runs that start with $\path$. 
To $\Mc$ we assign the probability space $
(\runs,\mathcal F,\pr)$, where $\runs$ is the set of all runs in $\Mc$, $\mathcal F$ is the $\sigma$-algebra generated by all $\mathsf{Cone}(\path)$,
and $\pr$ is the unique probability measure such that
$\pr[\mathsf{Cone}(s_0s_1\cdots s_k)] = 
\mu(s_0)\cdot\prod_{i=1}^{k} \Pm(s_{i-1},s_i)$, where the empty product equals $1$.
The respective expected value of a random variable $f:\runs\to\mathbb R$ is $\expected[f]=\int_\runs f\ d\,\pr$.

A non-empty set $\comp\subseteq \St$ of states is \emph{strongly connected} (SC) 
if for every $s,s'\in\comp$ there is a path from $s$ to $s'$.
A set of states $\comp\subseteq \St$ is a \emph{bottom strongly connected component} (BSCC) of $\Mc$, if it is a maximal SC, and
 for each $s\in \comp$ and $s'\in\St\setminus\comp$ we have $\Pm(s,s')=0$.
The sets of all SCs and BSCCs in $\Mc$ are denoted by $\scs$ and $\bscc$, respectively. 
Note that with probability $1$, the set of states that appear infinitely many times on a run forms a BSCC.
From now on, we use the standard notions of SC and BSCC for directed graphs as well.

\mybigspace
\section{Solution for reachability}
\label{sec:reach}

A fundamental problem in Markov chain verification is computing the probability that a certain set of goal states is reached.
For the rest of the paper, let $\Mc = (\St,\Pm,\init)$ be a Markov chain and  
$G \subseteq \St$ be the set of the goal states in~$\Mc$. We let 
$
\reach G:= \{ \run \in \runs \mid \exists i\geq 0: \run[i]\in G\}
$
denote the event that ``eventually a state in $G$ is reached.'' 
 %
The event $\reach G$ is measurable and its probability $\pr[\reach G]$ can be computed numerically 
or estimated using statistical algorithms. 
Since no bound on the number of steps for reaching $G$ is given, 
the major difficulty for any statistical approach is to decide how long each sampled path should be.
We can stop extending the path either when we reach $G$, or when no more new states can be reached anyways. 
The latter happens if and only if we are in a BSCC and we have seen all of its states.
%

In this section, we first show how to monitor each simulation run on the fly, in order to detect quickly if a BSCC has been entered with high probability.
Then, we show how to use hypothesis testing in order to estimate $\pr[\reach G]$. 

%

\mybigspace
\subsection{BSCC detection}

We start with an example illustrating how to measure probability of reaching a BSCC from one path observation.


\begin{example}\label{ex:motivate-cand}
Recall Example~\ref{fig:ex-MC-intro} and Figure~\ref{fig:ex-MC-intro}.
Now, consider an execution path $stuttutu$. 
Intuitively, does $\{t,u\}$ look as a good ``candidate'' for being a BSCC of $\Mc$? 
We visited both $t$ and $u$ three times; we have taken a transition from each $t$ and $u$ at least twice without leaving $\{t,u\}$.
By the same reasoning as in Example~\ref{fig:ex-MC-intro}, we could have missed some outgoing transition with probability at most $(1-\pmin)^2$. 
The structure of the system that can be deduced from this path is in Figure~\ref{fig:ev-MC}
and is correct with probability at least $1-(1-\pmin)^2$.
\qee\end{example}

Now we formalise our intuition.
Given a finite or infinite sequence $\run=s_0s_1\cdots$, the \emph{support} of $\run$ is the set $\support{\run}=\{s_0,s_1,\ldots\}$. Further, the \emph{graph of $\run$} is given by vertices $\support{\run}$ and edges $\{(s_i,s_{i+1})\mid i=0,1,\ldots\}$. 

\begin{figure}
\centering
\scalebox{0.8}{
\begin{tikzpicture}
\node[state,initial,initial text=] (s) at (0,1){$s$};
\node[state] (t) at (2,1){$t$};
\node[state] (u) at (4,1){$u$};
\path[->] 
(s) edge node[above]{} (t)
(t) edge[bend left] node[below]{} (u)
(u) edge[bend left] node[above]{} (t)
(t) edge[loop below] ()
;
\end{tikzpicture}
}

\caption{A graph of a path $stuttutu$.}\label{fig:ev-MC}
\end{figure}
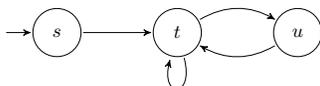

\begin{definition}[Candidate] 
If a path $\path$ has a suffix $\candidatepath$ such that $\support{\candidatepath}$ is a BSCC of the graph of $\path$,
we call $\support{\candidatepath}$ the \emph{candidate of $\path$}.
Moreover, for $k\in\Nset$, we call it a \emph{$k$-candidate (of $\path$)} if each $s\in\support{\candidatepath}$ has at least $k$ occurrences in $\candidatepath$ and the last element of $\candidatepath$ has at least $k+1$ occurrences.
A \emph{$k$-candidate of a run} $\run$ is a $k$-candidate of some prefix of $\run$.
\end{definition}
Note that for each path there is at most one candidate. 
Therefore, we write $\candidate(\path)$ to denote the candidate of $\path$ if there is one, and $\candidate(\path)=\bot$, otherwise.
Observe that each $\candidate(\path)\neq\bot$ is a SC in $\Mc$.

\begin{example}
Consider a path $\path=stuttutu$, then $\candidate(\path)=\{t,u\}$. Observe that $\{t\}$ is not a candidate as it is not maximal. 
Further, $\candidate(\path)$ is a $2$-candidate (and as such also a $1$-candidate), but not a $3$-candidate. 
Intuitively, the reason is that we only took a transition from $u$ (to the candidate) twice, cf.~Example~\ref{ex:motivate-cand}.
\qee\end{example}

Intuitively, the higher the $k$ the more it looks as if the $k$-candidate is indeed a BSCC.
Denoting by $\kcand k(K)$ the random predicate of $K$ being a $k$-candidate on a run, 
the probability of ``unluckily'' detecting any specific non-BSCC set of states $\candidate$ as a $k$-candidate, can be bounded as follows.

\begin{lemma}\label{lem:one-cand}
For every $\candidate\subseteq \St$ such that $\candidate\notin\bscc$, and every $s\in \candidate$,  $k\in\Nset$,
$$
\pr[\kcand k (\candidate) \mid \reach s]\leq (1-\pmin)^k\,.
$$
\end{lemma}
\begin{proof}
Since $\candidate$ is not a BSCC, there is a state $t\in\candidate$ with a transition to $t'\notin\candidate$. 
The set of states $\candidate$ is a $k$-candidate of a run, only if $t$ is visited at least $k$ times by the path and was never followed by $t'$ (indeed, even if $t$ is the last state in the path, by definition of a $k$-candidate, there are also at least $k$ previous occurrences of $t$ in the path). 
Further, since the transition from $t$ to $t'$ has probability at least $\pmin$, the probability of not taking the transition $k$ times is at most $(1-\pmin)^k$. 
\QED\end{proof}



\mybigspace
\begin{example}\label{ex:eigen}
We illustrate how candidates ``evolve over time'' along a run. 
Consider a run $\run=s_0s_0s_1s_0\cdots$ of the Markov chain in Figure~\ref{fig:self-looping states}. 
The empty and one-letter prefix do not have the candidate defined, $s_0s_0$ has a candidate $\{s_0\}$, then again $\candidate(s_0s_0s_1)=\bot$, and $\candidate(s_0s_0s_1s_0)=\{s_0,s_1\}$.
One can observe that subsequent candidates are either disjoint or contain some of the previous candidates. 
Consequently, there are at most $2|\St|-1$ candidates on every run, which is in our setting an unknown bound.
\qee\end{example}
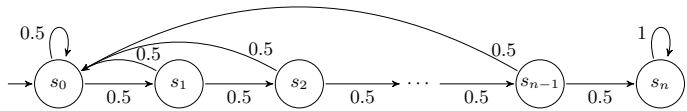
\begin{figure}[h]
\mybigspace
\centering
\scalebox{0.8}{
\begin{tikzpicture}
\node[state,initial,initial text=] (s) at (0,1){$s_0$};
\node[state] (t) at (2,1){$s_1$};
\node[state] (u) at (4,1){$s_2$};
\node (v) at (6,1){$\cdots$};
\node[state] (w) at (8,1){$s_{n-1}$};
\node[state] (x) at (10,1){$s_{n}$};
\path[->] 
(s) edge node[below]{$0.5$} (t)
(t) edge node[below]{$0.5$} (u)
(u) edge node[below]{$0.5$} (v)
(v) edge node[below]{$0.5$} (w)
(w) edge node[below]{$0.5$} (x)
(t) edge[bend right]  node[above, pos=0.1]{0.5} (s)
(u) edge[bend right]  node[above, pos=0.08]{0.5} (s)
(w) edge[bend right]  node[above, pos=0.03]{0.5} (s)
(s) edge[loop above] node[left,pos=0.3]{$0.5$} ()
(x) edge[loop above] node[left,pos=0.3]{$1$} ()
;
\end{tikzpicture}}
\caption{A family (for $n\in\Nset$) of Markov chains with large eigenvalues.}
\mybigspace
\label{fig:self-looping states}
\end{figure}

While we have bounded the probability of detecting any specific non-BSCC set $\candidate$ as a $k$-candidate, 
we need to bound the overall error for detecting a candidate that is not a BSCC.
Since there can be many false candidates on a run before the real BSCC (e.g. Figure~\ref{fig:self-looping states}), we need to bound the error of reporting any of them.
%

In the following, we first formalise the process of discovering candidates along the run. 
Second, we bound the error that any of the non-BSCC candidates becomes a $k$-candidate.
Third, we bound the overall error of not detecting the real BSCC by increasing $k$ every time a different candidate is found.

We start with discovering the sequence of candidates on a run. For a run $\run=s_0s_1\cdots$, consider the sequence of random variables defined by $\candidate(s_0\ldots s_j)$ for $j\geq 0$, and let $(\candidate_i)_{i\geq 1}$ be the subsequence without undefined elements and with no repetition of consecutive elements.
For example, for a run $\varrho=s_0s_1s_1s_1s_0s_1s_2s_2\cdots$, we have $K_1=\{s_1\}$, $K_2=\{s_0,s_1\}$, $K_3=\{s_2\}$, etc.
Let $\candidate_j$ be the last element of this sequence, called the \emph{final candidate}.
Additionally, we define $\candidate_\ell:=\candidate_j$ for all $\ell>j$. 
%
We describe the lifetime of a candidate.
Given a non-final $\candidate_i$, we write $\run=\alpha_i\beta_ib_i\gamma_id_i\delta_i$ so that 
$\support{\alpha_i}\cap \candidate_i=\emptyset$, 
$\support{\beta_ib_i\gamma_i}=\candidate_i$,
$d_i\notin \candidate_i$, and 
$\candidate(\alpha_i\beta_i)\neq \candidate_i$, $\candidate(\alpha_i\beta_ib_i)= \candidate_i$.
Intuitively, we start exploring $\candidate_i$ in $\beta_i$; $\candidate_i$ becomes a candidate in $b_i$, the birthday of the $i$th candidate; it remains to be a candidate until $d_i$, the death of the $i$th candidate. 
For example, for the run $\varrho=s_0s_1s_1s_1s_0s_1s_2s_2\cdots$ and $i=1$, $\alpha_1=s_0$, $\beta_1=s_1$, $b_1=s_1$, $\gamma_1=s_1$, $d_1=s_0$, $\delta_1=s_1s_2s_2\varrho[8]\varrho[9]\cdots$. 
Note that the final candidate is almost surely a BSCC of $\Mc$ and would thus have $\gamma_j$ infinite.

Now, we proceed to bounding errors for each candidate. Since there is an unknown number of candidates on a run, we will need a slightly stronger definition.
First, observe that $\kcand k(\candidate_i)$ iff $\candidate_i$ is a $k$-candidate of $\beta_ib_i\gamma_i$. We say $\candidate_i$ is a \emph{strong $k$-candidate}, written $\skcand k(\candidate_i)$, if it is a $k$-candidate of $b_i\gamma_i$.
Intuitively, it becomes a $k$-candidate even not counting the discovery phase.
As a result, even if we already assume there exists an $i$th candidate, its strong $k$-candidacy gives the guarantees on being a BSCC as above in Lemma~\ref{lem:one-cand}.

\begin{lemma}\label{lem:ith-cand}
For every $i,k\in\Nset$, we have $$\pr[\skcand k(\candidate_i) \mid \candidate_i\notin\bscc]\leq (1-\pmin)^k\,.$$ 
\end{lemma}
\begin{proof}
\begin{align*}
&\pr[\skcand k(\candidate_i) \mid \candidate_i\notin\bscc]\\
&= \frac1{\pr[\candidate_i\notin\bscc]} \sum_{\substack{C\in\scs\setminus\bscc\\s\in C}}\pr[\candidate_i=C,b_i=s]\cdot\pr[\skcand k(C) \mid \candidate_i=C, b_i=s]\\
&= \frac1{\pr[\candidate_i\notin\bscc]} \sum_{\substack{C\in\scs\setminus\bscc\\s\in C}}\pr[\candidate_i=C,b_i=s]\cdot\pr[\kcand k(C) \mid \reach s] \tag{by Markov property}\\
&\leq \frac1{\pr[\candidate_i\notin\bscc]} \sum_{\substack{C\in\scs\setminus\bscc\\s\in C}} \pr[\candidate_i=C,b_i=s]\cdot (1-\pmin)^k \tag{by Lemma~\ref{lem:one-cand}}\\
&= (1-\pmin)^k\tag{by $\pr[\candidate_i\in\scs,b_i\in K_i]=1$}
\end{align*}\QED
\end{proof}

Since the number of candidates can only be bounded with some knowledge of the state space, e.g.\ its size, we assume no bounds and provide a method to bound the error even for an unbounded number of candidates on a run.

\begin{lemma}\label{lem:more-cand}
For $(k_i)_{i=1}^\infty\in\Nset^\Nset$, let 
$\mathcal E\mathit{rr}$ be the set of runs such that for some $i\in\Nset$, we have $\skcand {k_i}(\candidate_i)$ despite $\candidate_i\notin\bscc$.
Then $$\displaystyle\pr[\mathcal E\mathit{rr}]<\sum_{i=1}^\infty (1-\pmin)^{k_i}\,.$$
\end{lemma}
\begin{proof}
\begin{align*}
\pr[\mathcal E\mathit{rr}]
&=\pr\left[\bigcup_{i=1}^{\infty} \Big( \skcand {k_i}(\candidate_i) \cap \candidate_i\notin\bscc\Big)\right]\\
&\leq\sum_{i=1}^{\infty}\pr[\skcand {k_i}(\candidate_i) \cap \candidate_i\notin\bscc] \tag{by the union bound}\\
&=\sum_{i=1}^{\infty}\pr[\skcand {k_i}(\candidate_i) \mid \candidate_i\notin\bscc]\cdot\pr[\candidate_i\notin\bscc]\\
&\leq\sum_{i=1}^{\infty}\pr[\skcand {k_i}(\candidate_i) \mid \candidate_i\notin\bscc]\\
&=\sum_{i=1}^{\infty}(1-\pmin)^{k_i} \tag{by Lemma~\ref{lem:ith-cand}}
\end{align*}
\QED
\end{proof}

In Algorithm~\ref{alg:cand} we present a procedure for deciding whether a BSCC inferred from a path $\path$ is indeed a BSCC with confidence greater than $1-\delta$.
We use notation $\textsc{SCand}_{k_i}(\candidate,\pi)$ to denote the function deciding whether $\candidate$ is a strong $k_i$-candidate on $\pi$. 
%
The overall error bound is obtained by setting $k_i=\frac{i-\log \delta}{-\log(1-\pmin)}$.
\mybigspace
\begin{algorithm}[t]
  \caption{\textsc{ReachedBSCC}}
  \label{alg:cand}
\begin{algorithmic}
\Require path $\path=s_0s_1\cdots s_n$, $\pmin,\delta\in(0,1]$
\Ensure $\True$ iff $\candidate(\pi)\in \bscc$
\State $C \gets \bot$, $i \gets 0$ 
\For {$j=0$ to $n$}
\If {$\candidate(s_0\cdots s_j) \neq \bot$ and $\candidate(s_0\cdots s_j) \neq C$ }
\State $C \gets \candidate(s_0\cdots s_j)$
\State $i \gets i+1$
\EndIf
\EndFor
\State $k_i\gets\frac{i-\log \delta}{-\log(1-\pmin)}$ 
\If {$i \geq 1$ and $\textsc{SCand}_{k_i}(\candidate(\pi),\pi)$}
\Return $\True$
\Else\ \Return $\False$
\EndIf
\end{algorithmic}
\end{algorithm}

\begin{theorem}\label{thm:main}
For every $\delta>0$, Algorithm~\ref{alg:cand} is correct with error probability at most $\delta$.
\end{theorem}
\begin{proof}
Since $M$ is finite, the Algorithm~\ref{alg:cand} terminates almost surely. 
The probability to return an incorrect result can be bounded by returning incorrect result for one of the non-final candidates, 
which by Lemma~\ref{lem:more-cand} is as follows: 
\[
\sum_{i=1}^\infty (1-\pmin)^{k_i}
= \sum_{i=1}^\infty (1-\pmin)^\frac{-i+\log \delta}{\log(1-\pmin)}
= \sum_{i=1}^\infty 2^{-i+\log \delta}
= \sum_{i=1}^\infty \delta/2^i
= \delta.
\]\QED
\end{proof}

We have shown how to detect a BSCC of a single path with desired confidence.
In Algorithm \ref{alg:reach}, we show how to use our BSCC detection method to decide whether a given path 
reaches the set $G$ with confidence $1-\delta$.
The function $\nextState(\path)$ randomly picks a state according to $\init$ if the path is empty ($\path=\emptypath$); otherwise, if $\ell$ is the last state of $\path$, it randomly chooses its successor according to $\Pm(\ell,\cdot)$.
The algorithm returns $\True$ when $\path$ reaches a state in $G$, 
and $\False$ when for some $i$, the $i$th candidate is a strong $k_i$-candidate.
In the latter case, with probability at least $1-\delta$, $\path$ has reached a BSCC not containing $G$.
%
%
Hence, with probability at most $\delta$, the algorithm returns $\False$ for a path that could reach a goal.
\mybigspace
\begin{algorithm}[b]
  \caption{\textsc{SinglePathReach}}
  \label{alg:reach}
\begin{algorithmic}
\Require goal states $G$ of $\Mc$, $\pmin,\delta\in(0,1]$
\Ensure $\True$ iff a run reaches $G$ 
\State $\path \gets \emptypath$
\Repeat
\State $s \gets \nextState(\path)$
\State $\path \gets \path \concat s$
\If  {$s\in G$}
\Return $\True$ \Comment We have provably reached $G$
\EndIf
\Until \textsc{ReachedBSCC}($\path,\pmin,\delta$) 
\State \Return $\False$ \Comment By Theorem~\ref{thm:main}, $\pr[\candidate(\path) \in \bscc] \geq 1 - \delta$
\end{algorithmic}
\end{algorithm}



\subsection{Hypothesis testing on a Bernoulli variable observed with bounded error}
\label{lab:ht}

In the following, we show how to estimate the probability of reaching a set of goal states, 
by combining the BSCC detection and hypothesis testing. 
More specifically, we sample many paths of a Markov chain, decide for each whether it reaches the goal states (Algorithm \ref{alg:reach}), and then use hypothesis testing to estimate the event probability. 
The hypothesis testing is adapted to the fact that testing reachability on a single path may report false negatives.

%

Let $X_\reach^\delta$ be a Bernoulli random variable, such that $X_\reach^\delta=1$ if and only if \textsc{SinglePathReach}$(G,\pmin,\delta)=\True$,
describing the outcome of Algorithm~\ref{alg:reach}.
The following theorem establishes that $X_\reach^\delta$ estimates $\pr[\reach G]$ with a bias bounded by $\delta$.

\begin{theorem}
\label{thm:reach_exp}
For every $\delta>0$, we have $\pr[\reach G]-\delta\leq\expected[X_\reach^\delta]\leq\pr[\reach G]$.
\end{theorem}
\begin{proof}
Since the event $\reach G$ is necessary for $X_\reach^\delta=1$, we have $\pr[\reach G\mid X_\reach^\delta=1]=1$.
Therefore, $\pr[X_\reach^\delta=1]=\pr[\reach G,X_\reach^\delta=1]\leq\pr[\reach G]$, hence the upper bound.
As for the lower bound, again
$\pr[X_\reach^\delta=1]=\pr[\reach G,X_\reach^\delta=1]=\pr[\reach G]-\pr[\reach G,X_\reach^\delta=0]\geq \pr[\reach G]-\delta$, where the last inequality follows by Theorem~\ref{thm:main} and the definition of BSCC.
%
\QED\end{proof}

%

In order to conclude on the value $\pr[\reach G]$, the standard statistical model checking approach via hypothesis testing~\cite{Younes02} decides between the hypothesis
$ H_0: \pr[\reach G]\geq p + \varepsilon$ and $H_1: \pr[\reach G]\leq p - \varepsilon$,
where $\varepsilon$ is a desired indifference region.
As we do not have precise observations on each path,
we reduce this problem to a hypothesis testing on the variable $X_\reach^\delta$ with a narrower indifference region: $ H_0': \expected[X_\reach^\delta] \geq p+(\varepsilon - \delta)$ and $H_1': \expected[X_\reach^\delta]\leq p - \varepsilon$, for some $\delta<\varepsilon$.

%
We define the reduction simply as follows. Given a statistical test $\mathcal{T'}$ for $H'_0, H'_1$ we define a test $\mathcal{T}$ that accepts $H_0$ if  $\mathcal{T}'$ accepts $H'_0$, and $H_1$ otherwise. 
The following lemma shows that $\mathcal{T}$ has the same strength as $\mathcal{T}'$.

\begin{lemma}
\label{lem:red_reach}
Suppose the test $\mathcal{T'}$  decides between $H_0'$ and $H_1'$ with strength $(\alpha, \beta)$.
Then the test $\mathcal{T}$ decides  between $H_0$ with $H_1$ with strength $(\alpha, \beta)$.
\end{lemma}
\begin{proof}
Consider type I error of $\mathcal{T}$.
Assume that $H_0$ holds, which means   $\pr[\reach G]\geq p + \varepsilon$.
By Theorem \ref{thm:reach_exp} it follows that  $\pr[X_\reach^\delta=1]\geq \pr[\reach G]-\delta \geq p+(\varepsilon - \delta)$, thus $H'_0$ also holds.
By assumption the test $\mathcal{T'}$ accepts $H'_1$ with probability at most $\alpha$, thus, by the reduction, $\mathcal{T}$ also accepts $H_1$ with probability $\leq \alpha$.
The proof for type II error is analogous.
\QED
\end{proof}

Lemma \ref{lem:red_reach} gives us the following algorithm to decide between $H_0$ and $H_1$.
We generate samples $x_0, x_1, \cdots, x_n \sim X_\reach^\delta$ from \textsc{SinglePathReach}$(G,\pmin,\delta)$, and apply a statistical test to decide between $H'_0$ and $H'_1$.
Finally, we accept $H_0$ if $H'_0$ was accepted by the test, and $H_1$ otherwise.
In our implementation, we used the sequential probability ration test (SPRT) \cite{YounesThesis,wald1945sequential} for hypothesis testing. 

\section{Extensions}
\label{sec:extens}

In this section, we present how the on-the-fly BSCC detection can be used for verifying LTL and quantitative properties (mean payoff).

\subsection{Linear temporal logic}

We show how our method extends to properties expressible by linear temporal logic (LTL)~\cite{pnueli1977temporal} and, in the same manner, to all $\omega$-regular properties. 
Given a finite set $Ap$ of atomic propositions,
a \emph{labelled Markov chain} (LMC) is a tuple $\Mc = (\St, \Pm, \init, Ap, \Lab)$, where $(\St,\Pm, \init)$ is a MC and $\Lab : \St \to 2^{Ap}$ is a labelling function. 
The formulae of LTL are given by the following syntax:
\begin{align*}
\varphi~::=~ &  a\mid \neg \varphi \mid \varphi \wedge \varphi  \mid \X\varphi \mid \varphi\U\varphi
\end{align*} 
for $a\in Ap$.
The semantics is defined with respect to a word $w\in (2^{Ap})^\omega$. 
The $i$th letter of $w$ is denoted by $w[i]$, i.e.~$w=w[0]w[1]\cdots$ and we write $w_i$ for the  suffix $w[i] w[i+1] \cdots.$
We define

\[
\begin{array}[t]{lclclcl}
  w \models a & \iff & a \in w[0] \\
  w \models \neg \varphi & \iff & \text{not } w \models \varphi \\
  w \models \varphi \wedge \psi & \iff & w \models \varphi \text{ and } w \models \psi\\
  w \models \X \varphi & \iff & w_1 \models \varphi \\
  w \models \varphi\U \psi & \iff &
\exists \, k\in\Nset: w_k \models \psi \text{ and } 
\forall\, 0\leq j < k: w_j\models \varphi
\end{array}
\]
The set $\{w\in (\draAl)^\omega\mid w\models\varphi\}$ is denoted by $\mathsf L(\varphi)$.

Given a labelled Markov chain $\Mc$ and an LTL formula $\varphi$, we are interested in the measure $\pr[\Mc \models \varphi] := \pr[\{\run\in\runs \mid \Lab(\run) \models \varphi\}],$
where $\Lab$ is naturally extended to runs by $\Lab(\run)[i]=\Lab(\run[i])$ for all $i$.  

For every LTL formula $\varphi$, one can construct a \emph{ deterministic Rabin automaton} (DRA) $\dra = (\draS, \draAl, \draTr, \draInit, \draAcc)$ that accepts all runs that satisfy $\varphi$~\cite{BK08}. 
Here $\draS$ is a finite set of states,
$\draTr : \draS \times \draAl \to \draS$ is the transition function,
$\draInit \in \draS$ is the initial state, and
$\draAcc \subseteq 2^\draS \times 2^\draS$ is the acceptance
condition.  A word $w\in (\draAl)^\omega$ induces an infinite sequence
$\dra(w)=s_0 s_1 \cdots\in \draS^\omega$, such that $s_0 = q_0$ and
$\draTr(s_i,w[i])= s_{i+1}$ for $i\geq 0$.  We write $\Inf(w)$ for the
set of states that occur infinitely often in $\dra(w)$.  Word $w$ is
accepted, if there exists a pair $(E,F)\in \draAcc$, such that
$E\cap \Inf(w) = \emptyset$ and $F\cap \Inf(w) \neq \emptyset$. The
language $\mathsf L(\dra)$ of $\dra$ is the set of all words accepted
by $\dra$. The following is a well known result, see e.g.~\cite{BK08}.

\begin{lemma}
For every LTL formula $\varphi$, a DRA $\dra$ can be effectively constructed such that $\mathsf L(\dra)=\mathsf L(\varphi)$.
\end{lemma}

Further, the product of 
a MC $\Mc$ and DRA $\mathcal A$ is the Markov chain $\mathcal{M}\otimes \mathcal{A}=(\St \times \draS, \Pm', \init')$, where $\Pm'((s,q),(s',q')) = \Pm(s,s')$ if $q'=\draTr(q,\Lab(s'))$ and $\Pm'((s,q),(s',q'))=0$ otherwise,  and $\init'(s,q) = \init(s)$ if $\draTr(\draInit, \Lab(s))=q$ and $\init'(s,q)=0$ otherwise. 
Note that $\mathcal{M}\otimes \mathcal{A}$ has the same smallest transition probability $\pmin$ as $\mathcal{M}$.

The crux of LTL probabilistic model checking relies on the fact that the probability of satisfying an LTL property $\varphi$ in a Markov chain $\Mc$ equals the probability of reaching an accepting BSCC in the Markov chain $\Mc \otimes \mathcal{A}_\varphi$. 
Formally, a BSCC $C$ of $\Mc \otimes \mathcal{A}_\varphi$ is \emph{accepting} if for some $(E,F)\in\draAcc$ we have $C\cap (S\times E)=\emptyset$ and $C\cap (S\times F)\neq\emptyset$. Let $\mathsf{AccBSCC}$ denote the union of all accepting BSCCs in $\Mc$. Then we obtain the following well-known fact \cite{BK08}:

\begin{lemma}
For every labelled Markov chain $\Mc$ and LTL formula $\varphi$, we have 
$\pr[\Mc \models \varphi] = \pr[\reach \mathsf{AccBSCC}]$.
\end{lemma}
\mybigspace
\begin{algorithm}[H]
  \caption{\textsc{SinglePathLTL}}
  \label{alg:ltl}
\begin{algorithmic}
\Require DRA $\dra = (\draS, \draAl, \draTr, \draInit, \draAcc)$, $\pmin,\delta\in(0,1]$
\Ensure $\True$ iff the final candidate is an accepting BSCC
\State $q\gets\draInit$, $\path \gets \emptypath$ 
\Repeat
\State $s \gets \nextState(\path)$
\State $q \gets \draTr(q, \Lab(s))$
\State $\path \gets \path \concat (s,q)$
\Until \textsc{ReachedBSCC}($\path,\pmin,\delta$)  \Comment $\pr[\candidate(\path) \in \bscc] \geq 1 - \delta$
\State \Return $\exists (E,F)\in \draAcc : \candidate(\path)\cap (S\times E) = \emptyset \land \candidate(\path)\cap (S\times F) \neq \emptyset$  
\end{algorithmic}
\end{algorithm}
\mybigspace

Since the input used is a Rabin automaton, the method applies to all $\omega$-regular properties.
Let $X^\delta_\varphi$ be a Bernoulli random variable, such that $X_\varphi^\delta=1$ if and only if \textsc{SinglePathLTL}$(\dra_\varphi,\pmin,\delta)=\True$.
Since the BSCC must be reached and fully explored to classify it correctly, the error of the algorithm can now be both-sided.

\begin{theorem}
For every $\delta>0$, $\pr[\Mc\models\varphi] - \delta \leq \expected[X_\varphi^\delta] \leq \pr[\Mc\models\varphi] + \delta$.
\end{theorem}
Further, like in Section \ref{lab:ht}, we can reduce the hypothesis testing problem for 
\[H_0: \pr[\Mc \models \varphi]\geq p + \varepsilon\text{~\qquad and ~\qquad}H_1: \pr[\Mc \models \varphi]\leq p - \varepsilon\] 
for any $\delta<\varepsilon$ to the following hypothesis testing problem on the observable $X^\delta_\varphi$
\[H'_0: \expected[X^\delta_\varphi]\geq p + (\varepsilon-\delta)\text{\qquad and \qquad}H'_1: \expected[X^\delta_\varphi]\leq p - (\varepsilon - \delta)\,.\] 

\mybigspace\mybigspace

\subsection{Mean payoff}

We show that our method extends also to quantitative properties, such as mean payoff (also called long-run average reward).
%
Let $\Mc = (S,\Pm,\init)$ be a Markov chain and $r:S\rightarrow [0,1]$ 
be a \emph{reward} function. 
%
%
Denoting by $S_i$ the random variable returning the $i$-th state on a run, the aim is to compute
$$\MP:=\lim_{n\to\infty} \expected\left[\frac1n\sum_{i=1}^n r(S_i)\right]\,.$$
This limit exists (see, e.g.~\cite{norris1998markov}), and equals $\sum_{C\in\bscc}\pr[\reach C]\cdot \MP_C,$
where $\MP_C$ 
is the mean payoff of runs ending in $C$. 
Note that $\MP_C$ can be computed from $r$ and transition probabilities in $C$~\cite{norris1998markov}.
We have already shown how our method estimates $\pr[\reach C]$. Now we show how it extends to estimating transition probabilities in BSCCs and thus the mean payoff.
%

%
First, we focus on a single path $\path$ that has reached a BSCC $C=\candidate(\pi)$ and show how to estimate the transition probabilities $\Pm(s,s')$ for each $s,s'\in {\ccand}$.
Let $X_{s,s'}$ be the random variable denoting the event that $\nextState(s)=s'$. 
$X_{s,s'}$ is a Bernoulli variable with parameter $\Pm(s,s')$, so we use the obvious estimator $\hat\Pm(s,s')=\#_{ss'}(\path)/\#_{s}(\path)$, where $\#_{\alpha}(\pi)$ is the number of occurrences of $\alpha$ in $\pi$.
If $\pi$ is long enough so that $\#_s(\path)$ is large enough, the estimation is guaranteed to have desired precision $\trerr{s,s'}$ with desired confidence $(1-\delta_{s,s'})$. Indeed,  using H\"{o}ffding's inequality, we obtain
\begin{equation}
\pr[\hat\Pm(s,s') - \Pm(s,s')|>\trerr{s,s'}] \leq \delta_{s,s'}=2e^{-2 \#_s(\path)\cdot\trerr{s,s'}^2}\,.
\label{eq:trerr}
\end{equation}
Hence, we can extend the path $\pi$ with candidate $C$ until it is long enough so that we have a $1-\delta_{\ccand}$ confidence that {all} the transition probabilities in $\ccand$ are in the $\trerr{}$-neighbourhood of our estimates, by ensuring that
%
$\sum_{s,s'\in C}\delta_{s,s'}<\delta_{\ccand}$.
These estimated transition probabilities $\hat\Pm$ induce a mean payoff $\hat\MP_{\ccand}$. 
Moreover, $\hat\MP_{\ccand}$ estimates the real mean payoff $\MP_{\ccand}$.
Indeed, by \cite{chatterjee2012robustness,solan2003continuity}, 
\begin{equation}
|\hat\MP_{\ccand}-\MP_{\ccand}|\leq \mperr{}:=\left(1+\frac{\trerr{}}{\pmin}\right)^{2\cdot|{\ccand}|}-1\,. \label{eq:mperr}
\end{equation}
Note that by Taylor's expansion, for small $\trerr{}$, we have $\mperr{}\approx 2|{\ccand}|\trerr{}$.

\mybigspace
\begin{algorithm}[H]
  \caption{\textsc{SinglePathMP}}
  \label{alg:meanpayoff}
\begin{algorithmic}
\Require reward function $r$, $\pmin,\mperr{},\delta\in(0,1]$,
\Ensure $\hat\MP_C$ such that $|\hat\MP_C-\MP_C|<\mperr{}$ where $C$ is the BSCC of the generated run
\State $\path \gets \emptypath$
\Repeat
\State $\path \gets \path \,.\, \nextState(\path)$
\If {$\candidate(\path) \neq \bot$}
\State $\trerr{}=\pmin((1+\mperr{})^{1/2|\candidate(\path)|}-1)$ \Comment By Equation (\ref{eq:mperr})
\State $k \gets \frac{\ln(2 |\candidate(\path)|^2)- \ln (\delta/2)}{2 \trerr{}^2 }$ \Comment By Equation~(\ref{eq:trerr}) 
\EndIf
\Until \textsc{ReachedBSCC}($\path,\pmin,\delta/2$) and $\textsc{SCand}_{k}(\candidate(\path),\path)$ 
\State \Return $\hat\MP_{\candidate(\path)}$ computed from $\hat\Pm$ and $r$
\end{algorithmic}
\end{algorithm}
\mybigspace

Algorithm~\ref{alg:meanpayoff} extends Algorithm~\ref{alg:reach} as follows. 
It divides the confidence parameters $\delta$ into $\delta_{BSCC}$ (used as in Algorithm~\ref{alg:reach} to detect the BSCC) and $\delta_C$ (the total confidence for the estimates on transition probabilities). For simplicity, we set $\delta_{BSCC}=\delta_{C}=\delta/2$.
First, we compute the bound $\trerr{}$ required for $\mperr{}$-precision (by Eq.~\ref{eq:mperr}). 
Subsequently, we compute the required strength $k$ of the candidate guaranteeing $\delta_{\ccand}$-confidence on $\hat\Pm$ (from Eq.~\ref{eq:trerr}).
The path is prolonged until the candidate is strong enough; in such a case $\hat\MP_{\ccand}$ is $\mperr{}$-approximated with $1-\delta_{\ccand}$ confidence. 
If the candidate of the path changes, all values are computed from scratch for the new candidate.

\begin{theorem}\label{thm:mp}
For every $\delta>0$, the Algorithm \ref{alg:meanpayoff} terminates correctly with probability at least $1-\delta$.
\end{theorem} 
\begin{proof}
From Eq.~\ref{eq:trerr}, by the union bound, we are guaranteed that the probability that \emph{none} of the estimates $\hat{\Pm}_{s,s'}$ is outside of the $\mperr{}$-neighbourhood doesn't exceed the sum of all respective estimation errors, that is, $\delta_C=\sum_{s,s'\in\ccand} \delta_{s,s'}$. 
Next, from Eq.~\ref{eq:mperr} and from the fact that $C$ is subject to Theorem~\ref{thm:main} with confidence $\delta_{BSCC}$,
\begin{align*}
P(|\MP_C(r)&-\hat{\MP}_C(r)|>\mperr{}) = \\
= & P(C\in\bscc)P(|\MP(r)-\hat{\MP}(r)|>\mperr{}\mid C \in\bscc) +\\
& P(C\notin\bscc)P(|\MP(r)-\hat{\MP}(r)|>\mperr{}\mid C\notin\bscc)\\
\leq & 1\cdot\delta_C+\delta_{BSCC}\cdot 1 = \delta_C+\delta_{BSCC}\leq \delta.
\end{align*}
\QED\end{proof}


Let random variable $X_\MP^{\mperr{},\delta}$ denote the value \textsc{SinglePathMP}($r,\pmin,\mperr{},\delta$).
The following theorem establishes relation between the mean-payoff $\MP$ and the expected value of $X_\MP^{\mperr{},\delta}$.

\begin{theorem}
\label{thm:mp_conf}
  For every $\delta, \mperr{} > 0$,  $\MP - \mperr{} - \delta \leq \expected[X_\MP^{\mperr{},\delta}] \leq \MP + \mperr{} + \delta$.
\end{theorem}
\begin{proof}
Let us write $X_\MP^{\mperr{},\delta}$ as an expression of  random variables $Y,W,Z$
\[
X^{\mperr{},\delta}_\MP \;=\; Y(1-W) + WZ,  
\]
where 1) $W$ is a Bernoulli random variable, such that $W=0$ iff the algorithm correctly detected the BSCC and estimated transition probabilities within bounds,
2) $Y$ is the value computed by the algorithm if $W=0$, and the real mean payoff $\MP$ when $W=1$,
and 3) $Z$ is any random variable with the range $[0,1]$.
The interpretation is as follows: when $W=0$ we observe the result $Y$, which has bounded error $\mperr{}$, and when $W=1$ we observe arbitrary $Z$.
We note that  $Y,W,Z$ are not necessarily independent.
By Theorem \ref{thm:mp} $\expected[W] \leq \delta$ and by linearity of expectation: $\expected[X_\MP^{\mperr{},\delta}] = \expected[Y] - \expected[YW] + \expected[WZ]$.
For the upper bound, observe that $\expected[Y] \leq \MP + \mperr{}$, $\expected[YW]$ is non-negative and $ \expected[WZ] \leq \delta$.
As for the lower bound, note that $\expected[Y] \geq \MP - \mperr{}$,  $\expected[YW] \leq \delta$ and $\expected[WZ]$ is non-negative.
\QED \end{proof}
As a consequence of Theorem \ref{thm:mp_conf}, if we establish that with $(1-\alpha)$ confidence $X_\MP^{\mperr{},\delta}$ belongs to the interval $[a,b]$, then we can conclude with $(1-\alpha)$ confidence that $\MP$ belongs to the interval $[a-\mperr{} - \delta, b+\mperr{}+\delta]$.
Standard statistical methods can be applied to find the confidence bound for  $X_\MP^{\mperr{},\delta}$.



\newcommand{\emc}{\textsf{MC}}
\newcommand{\ecand}{\mathsf{SimAdaptive}}
\newcommand{\eterm}{\mathsf{SimTermination}}
\newcommand{\eanal}{\mathsf{SimAnalysis}}
\newcommand{\tout}{\mathsf{TO}}
\renewcommand{\wr}{\mathsf{WRONG}}
\newcommand{\memout}{\mathsf{MO}}
\newcommand{\pl}{\mathsf{PL}}
\newcommand{\error}{\mathsf{ERR}}

\newcommand{\ba}[1]{\cellcolor[gray]{0.6}{#1}}
\newcommand{\bb}[1]{#1}

\mybigspace
\section{Experimental evaluation}
\label{sec:exper}

We implemented our algorithms in the  probabilistic model checker \textsc{Prism} \cite{prism}, and evaluated them on the DTMC examples from the \textsc{Prism} benchmark suite \cite{benchmarks}.
The benchmarks model  communication and security protocols, distributed algorithms, and fault-tolerant systems.
To demonstrate how our method performs depending on the topology of Markov chains, we also performed experiments on the generic  DTMCs shown in Figure \ref{fig:self-looping states} and Figure \ref{fig:add}, as well as on two  CTMCs from the literature that have large BSCCs: ``tandem'' \cite{hermanns1999multi} and ``gridworld'' \cite{DBLP:journals/sttt/YounesKNP06}.

\begin{table}[H]
\caption{Experimental results for unbounded reachability. Simulation parameters: $\alpha=\beta=\varepsilon=0.01$, $\delta=0.001$, $\pt=0.0001$. 
$\tout$ means time-out, and $\memout$ means memory-out.
 Our approach is denoted by $\ecand$ here.
Highlights show the best result the among topology-agnostic methods. }
\label{tab:results}
\centering
\resizebox{\columnwidth}{!}{
\begin{tabular}{cccc|c|c||cc|c} 
\hline
\multicolumn{3}{c}{Example} & BSCC  & \multicolumn{1}{c|}{$\ecand$} & \multicolumn{1}{c||}{$\eterm$\cite{sbmf11}} & \multicolumn{2}{c|}{$\eanal$\cite{sbmf11}} & \multicolumn{1}{c}{$\emc$}\\
name & size & $\pmin$ & no., max.\ size  &		time & 	time &  	time  & analysis &	time\\
\hline
bluetooth(4)                              & 149K                  & $7.8\cdot10^{-3}$     & 3K, 1                 & \ba{  2.6s} &       16.4s    &      83.2s    &     80.4s    &  78.2s   \\
bluetooth(7)                              & 569K                  & $7.8\cdot10^{-3}$     & 5.8K, 1               & \ba{  3.8s} &       50.2s    &     284.4s    &    281.1s    & 261.2s   \\
bluetooth(10)                             & >569K                 & $7.8\cdot10^{-3}$     & >5.8K, 1              & \ba{  5.0s} &      109.2s    &   $\tout$     &    -         &    $\tout$       \\
\hline
  brp(500,500)                              & 4.5M                  & $0.01$                & 1.5K, 1             & \ba{  7.6s} &    13.8s    	&     35.6s 	 &     30.7s    &   103.0s      \\
brp(2K,2K)                                & 40M                   & $0.01$                & 4.5K, 1               &    20.4s    &   \ba{ 17.2s} &    824.4s      & 789.9s       & $\tout$       \\
brp(10K,10K)                              & >40M                  & $0.01$                & >4.5K, 1              &    89.2s    &   \ba{ 15.8s} &    $\tout$     &   -          &    $\tout$       \\
\hline
crowds(6,15)                              & 7.3M                  & $0.066$               & >3K, 1                & \ba{  3.6s} &   253.2s    &   \bb{2.0s} &     0.7s     &    19.4s      \\
crowds(7,20)                              & 17M                   & $0.05$                & >3K, 1                & \ba{  4.0s} &   283.8s    &   \bb{2.6s} &      1.1s    &   347.8s      \\
crowds(8,20)                              & 68M                   & $0.05$                & >3K, 1                & \ba{  5.6s} &   340.0s    &   \bb{4.0s} &      1.9s    & $\tout$       \\
\hline
eql(15,10)                                & 616G                  & $0.5$                 & 1, 1                  & \ba{ 16.2s} &  $\tout$     &      151.8s    & 145.1s    & 110.4s   \\
eql(20,15)                                & 1279T                 & $0.5$                 & 1, 1                  & \ba{ 28.8s} &  $\tout$     &      762.6s    & 745.4s    & 606.6s   \\
eql(20,20)                                & 1719T                 & $0.5$                 & 1, 1                  & \ba{ 31.4s} &  $\tout$     &    $\tout$     & -         &      $\tout$       \\
\hline
herman(17)                                & 129M                  & $7.6\cdot10^{-6}$     & 1, 34                 & \ba{ 23.0s} &       33.6s    &    21.6s      &      0.1s    & \bb{  1.2s}   \\
herman(19)                                & 1162M                 & $1.9\cdot10^{-6}$     & 1, 38                 & \ba{ 96.8s} &      134.0s    &    86.2s      &      0.1s    & \bb{  1.2s}   \\
herman(21)                                & 10G                   & $4.7\cdot10^{-7}$     & 1, 42                 & \ba{570.0s} &    $\tout$     &      505.2s   &      0.1s    & \bb{  1.4s}   \\
\hline
leader(6,6)                               & 280K                  & $2.1\cdot10^{-5}$     & 1, 1                  & \ba{  5.0s} &        5.4s    &      536.6s    &    530.3s    & 491.4s   \\
leader(6,8)                               & >280K                 & $3.8\cdot10^{-6}$     & 1, 1                  & \ba{ 23.0s} &       26.0s    & $\memout$      &  -           &     $\memout$   \\
leader(6,11)                              & >280K                 & $5.6\cdot10^{-7}$     & 1, 1                  & \ba{153.0s} &      174.8s    & $\memout$      &  -           &     $\memout$   \\
\hline
nand(50,3)                                & 11M                   & $0.02$                & 51, 1                 & \ba{  7.0s} &   231.2s    &  36.2s &  31.0s    &   272.0s      \\
nand(60,4)                                & 29M                   & $0.02$                & 61, 1                 & \ba{  6.0s} &   275.2s    &  60.2s &  56.3s    & $\tout$       \\
nand(70,5)                                & 67M                   & $0.02$                & 71, 1                 & \ba{  6.8s} &   370.2s    & 148.2s & 144.2s    & $\tout$       \\
\hline
tandem(500)                                & >1.7M                 & $2.4\cdot10^{-5}$     & 1, >501K              & \ba{  2.4s}  &        6.4s    &     4.6s    &      3.0s    &   3.4s   \\
tandem(1K)                               & 1.7M                  & $9.9\cdot10^{-5}$     & 1, 501K                 & \ba{  2.6s}  &       19.2s    &    17.0s    &     12.7s    &  13.0s   \\
tandem(2K)                                & >1.7M                 & $4.9\cdot10^{-5}$     & 1, >501K               & \ba{  3.4s}  &       72.4s    &    62.4s    &     59.8s    &  59.4s   \\
\hline
gridworld(300)                            & 162M                  & $1\cdot10^{-3}$       & 598, 89K              & \ba{  8.2s} &    81.6s    & $\memout$   &   - &    $\memout$     \\
gridworld(400)                            & 384M                  & $1\cdot10^{-3}$       & 798, 160K             & \ba{  8.4s} &     100.6s  & $\memout$   &   - &    $\memout$     \\
gridworld(500)                            & 750M                  & $1\cdot10^{-3}$       & 998, 250K             & \ba{  5.8s} &    109.4s   & $\memout$   &   - &    $\memout$     \\
\hline
Fig.\ref{fig:self-looping states}(16)     & 37                    & $0.5$                 & 1, 1                  & \ba{58.6s} &   $\tout$    & 23.4s 		& 0.4s      & 2.0s\\
Fig.\ref{fig:self-looping states}(18)     & 39                    & $0.5$                 & 1, 1                  & $\tout$    &  $\tout$     & 74.8.0s 	& 1.8s      & 2.0s\\
Fig.\ref{fig:self-looping states}(20)     & 41                    & $0.5$                 & 1, 1                  & $\tout$    &   $\tout$    &   513.6s    	&  11.3s    &     2.0s      \\
\hline
Fig.\ref{fig:add}(1K,5)                   & 4022                  & $0.5$                 & 2, 5                  & \ba{  7.8s} &     218.2s    &     3.2s      &   0.5s    & \bb{  1.2s}   \\
Fig.\ref{fig:add}(1K,50)                  & 4202                  & $0.5$                 & 2, 50                 & \ba{ 12.4s} &     211.8s    &     3.6s      &   0.7s    & \bb{  1.0s}   \\
Fig.\ref{fig:add}(1K,500)                 & 6002                  & $0.5$                 & 2, 500,               &   431.0s    &   \ba{218.6s} &     3.6s      &   1.0s    & \bb{  1.2s}   \\
Fig.\ref{fig:add}(10K,5)                  & 40K                   & $0.5$                 & 2, 5                  & \ba{ 52.2s} &   $\tout$     &       42.2s   &  25.4s    & \bb{ 25.6s}   \\
Fig.\ref{fig:add}(100K,5)                 & 400K                  & $0.5$                 & 2, 5                  &   604.2s    &   \ba{  5.4s} &    $\tout$    & - 	    &      $\tout$       \\
\hline
\end{tabular}
}
\end{table}


%
All benchmarks are parametrised by one or more values, which influence their size and complexity, e.g.\ the number of modelled components.
We have made minor modifications to the benchmarks that could not be handled directly by the SMC component of \textsc{Prism}, by adding self-loops to deadlock states and fixing one initial state instead of multiple.
%
%
%
%

Our tool can be downloaded at \cite{tool}.
Experiments were done on a Linux 64-bit machine running an AMD Opteron $6134$
CPU  with a time limit of 15 minutes and a memory limit of 5GB.
To increase performance of our tool, we check whether a candidate has been found every $1000$ steps; this optimization does not violate correctness of our analysis. See Appendix \ref{app:id} for a discussion on this bound.

\begin{figure}[b]
\mybigspace
\centering
\scalebox{0.8}{
\begin{tikzpicture}
\node[state,initial,initial above, initial text=] (s_0) at (0,0){$s$};
\node[state,] (s_10) at (1.5,0){$u_1$};
\node[state,] (s_19) at (-1.5,0){$t_1$};

\node
(s_20)at (3,0) {$\cdots$};
\node[state] (s_30)at (4.5,0) {$u_N$};
\node[rectangle, draw] (s_40)at (6,0) {BSCC};

\node
(s_29)at (-3,0) {$\cdots$};
\node[state] (s_39)at (-4.5,0) {$t_N$};
\node[rectangle, draw] (s_49)at (-6,0) {BSCC};

\path[->]
(s_0) edge node[above] {$0.5$} (s_10)
(s_0) edge node[above] {$0.5$} (s_19)

(s_10) edge[loop above] node[left,pos=0.2] {$0.5$} (s_10)
(s_10) edge node[above] {$0.5$} (s_20)
(s_20) edge[dashed] (s_30)
(s_30) edge[loop above] node[left,pos=0.2] {$0.5$} (s_40)
(s_30) edge node[above] {$0.5$} (s_40)

(s_19) edge[loop above] node[left,pos=0.2] {$0.5$} (s_19)
(s_19) edge node[above] {$0.5$} (s_29)
(s_29) edge[dashed] (s_39)
(s_39) edge[loop above] node[left,pos=0.2] {$0.5$} (s_49)
(s_39) edge node[above] {$0.5$} (s_49);

\end{tikzpicture}}
\caption{A Markov chain with two transient parts consisting of $N$ strongly connected singletons, leading to  BSCCs with the ring topology of $M$ states.}

\label{fig:add}
\end{figure}
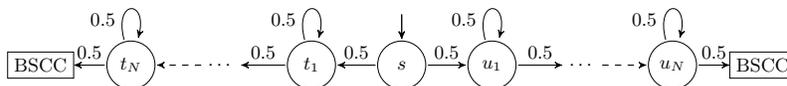

\para{Reachability}
The experimental results for unbounded reachability are shown in Table \ref{tab:results}. 
The \textsc{Prism} benchmarks were checked against their standard properties, when available.
%
We directly compare our method to another topology-agnostic method  of \cite{sbmf11} ($\eterm$), where at every step the sampled path is terminated with probability $\pt$.
The approach of \cite{atva14} with a priori bounds is not included, since it times out even on the smallest benchmarks. 
In addition, we performed experiments on two methods that are topology-aware: sampling with reachability analysis of \cite{sbmf11} ($\eanal$) and the numerical model-checking algorithm of \textsc{Prism} ($\emc$).
Appendix \ref{app:de}  contains detailed experimental evaluation of these methods.

%
The table shows the size of every example, its minimum probability, the number of BSCCs, and the size of the largest BSCC.
Column ``time'' reports the total wall time for the respective algorithm,
 and ``analysis'' shows the time for symbolic reachability analysis in the $\eanal$ method.
Highlights show the best result among the topology-agnostic methods.
All statistical methods were used with the SPRT test for choosing between the hypothesis, and their results were averaged over five runs.

Finding the optimal termination probability $\pt$ for the $\eterm$ method is a non-trivial task.
If the probability is too high, the method might never reach the target states, thus give an incorrect result, and if the value is too low, then it might sample unnecessarily long traces that never reach the target.
For instance, to ensure a correct answer on the Markov chain in Figure \ref{fig:self-looping states}, $\pt$ has to decrease exponentially with the number of states.
By experimenting we found that the probability $\pt=0.0001$ is low enough to ensure correct results.
See Appendix \ref{app:de} for experiments with other values of $\pt$.

On most examples, our method scales better than the $\eterm$ method.
Our method performs well even on examples with large BSCCs, such as ``tandem'' and ``gridworld,'' due to early termination when a goal state is reached.
For instance, on the ``gridworld'' example, most BSCCs do not contain a goal state, thus have to be fully explored, however the probability of reaching such BSCC is low, and as a consequence full BSCC exploration rarely occurs.
The $\eterm$ method performs well when the target states are unreachable or can be reached by short paths.
When long paths are necessary to reach the target, the probability that an individual path reaches the target is small, hence many samples are necessary to estimate the real probability with high confidence.

Moreover, it turns out that our method compares well even with methods that have access to the topology of the system. 
In many cases, the running time of the numerical algorithm $\emc$ increases dramatically with the size of the system, while  remaining almost constant in our method.
The bottleneck of the $\eanal$ algorithm is the reachability analysis of states that cannot reach the target, which in practice can be as difficult as numerical model checking.


\para{LTL and mean payoff}
In the second experiment, we compared our algorithm for checking LTL properties and estimating the mean payoff with the numerical methods of \textsc{Prism}; the results are shown in  Table \ref{tab:results_ltl}.
We compare against \textsc{Prism}, since we are not aware of any SMC-based or topology-agnostic approach for mean payoff, or full LTL.
For mean payoff, we computed  $95\%$-confidence bound of size $0.22$ with parameters $\delta=0.011, \mperr{}=0.08$, and for LTL we used the same parameters as for reachability.
Due to space limitations, we report results only on some models of each type, where either method did not time out.
%
In general our method scales better when  BSCCs are fairly small and are discovered quickly. 

\begin{table}[H]
\mybigspace
\caption{Experimental results for LTL and mean-payoff properties. Simulation parameters for LTL: $\alpha=\beta=\varepsilon=0.01$, $\delta=0.001$, for mean-payoff we computed $95\%$-confidence bound of size $0.22$ with  $\delta=0.011, \mperr{}=0.08$.}
\label{tab:results_ltl}
\centering
\resizebox{\columnwidth}{!}{
\begin{tabular}{cccc|ccc}
\hline
\multicolumn{4}{c|}{LTL} & \multicolumn{3}{c}{Mean payoff } \\
name & 	property	 $\quad$& $\ecand$	time$\quad$ & $\emc$ time$\quad$  & name & $\ecand$ time$\quad$ & $\emc$ time   \\

\hline
$\quad$bluetooth(10)$\quad$ & $\alws\reach$ & 	 		\ba{8.0s}    & $\tout$     & 		$\quad$bluetooth(10)$\quad$	&	\ba{3.0s}    & $\tout$       \\
brp(10K,10K) & 	$\reach\alws$ & 				\ba{90.0s}    & $\tout$     & 		brp(10K,10K)    & 	\ba{6.6s}    & $\tout$       \\
crowds(8,20) &	$\reach\alws$ & 				\ba{9.0s}    & $\tout$     & 		crowds(8,20)	& 	\ba{2.0s}    & $\tout$       \\
eql(20,20) & 	$\alws\reach$ & 				\ba{7.0s}    & $\memout$   & 		eql(20,20)      &	\ba{2.6s}    & $\tout$       \\
herman(21)  &	$\alws\reach$ & 				$\tout$     &     \ba{2.0s}    & 	herman(21)      &	$\memout$   &     \ba{3.0s}      \\
leader(6,5) &	$\alws\reach$ & 				277.0s    &   \ba{117.0s}    & 		leader(6,6)     &	\ba{48.5}	& 576.0 \\
nand(70,5) & 	$\alws\reach$ & 				\ba{4.0s}    & $\tout$     & 		nand(70,5)      & 	\ba{2.0s}    &   294.0s      \\
tandem(2K) &   $\alws\reach$ & 					$\tout$     &   \ba{221.0s}    & 	tandem(500)           		& $\tout$     &   \ba{191.0s}      \\
gridworld(100) & $\alws\reach \rightarrow \reach\alws$ & 	$\tout$     &   \ba{110.4s}   &     	gridworld(50)			& $\tout$ 	& \ba{58.1s} \\
Fig.\ref{fig:self-looping states}(20) &  $\alws\reach \rightarrow \alws\reach$ & $\tout$     &  3.4 & 	Fig.\ref{fig:self-looping states}(20) & $\tout$ & \ba{1.8s} \\
Fig.\ref{fig:add}(100K,5) &    $\alws\reach$ &			\ba{348.0s}    & $\tout$     & 		Fig.\ref{fig:add}(100K,5)	& 	\ba{79.6s}    & $\tout$       \\
Fig.\ref{fig:add}(1K,500) &   $\alws\reach$ & 			827.0s    &     \ba{2.0s}    & 		Fig.\ref{fig:add}(1K,500)       & $\tout$    &    \ba{2.0s}      \\
\hline
\end{tabular}
 }\mybigspace
\end{table}



\newcommand{\depth}{R}
\newcommand{\bsccsize}{B}
\newcommand{\bscctime}{T}
\newcommand{\simnum}{\mathit{sim}}

\mybigspace
\section{Discussion and conclusion}
\label{sec:discus}

As demonstrated by the experimental results, our method is fast on systems that are (1) shallow, and (2) with small BSCCs.
In such systems, the BSCC is reached quickly and the candidate is built-up quickly.
Further, recall that the BSCC is reported when a $k$-candidate is found, and that $k$ is increased with each candidate along the path.
Hence, when there are many strongly connected sets, and thus many candidates, the BSCC is detected by a $k$-candidate for a large $k$. 
However, since $k$ grows linearly in the number of candidates, the most important and limiting factor is the size of BSCCs.

We state the dependency on the depth of the system and BSCC sizes formally. We pick $\delta:=\frac\varepsilon2$ and let \mybigspace

$$
\simnum = \frac{ -\log\frac{\beta}{1-\alpha}\log\frac{1-\beta}{\alpha} }{\log\frac{p-\varepsilon + \delta }{p + \varepsilon - \delta}\log\frac{1-p - \varepsilon + \delta}{1-p + \varepsilon - \delta}}
\qquad \text{and} \qquad
k_{i}=\frac{i-\log\delta}{-\log(1-\pmin)}
$$ 
denote the a priori upper bound on the number of simulations necessary for SPRT~\cite{YounesThesis,wald1945sequential} and the strength of candidates as in 
Algorithm~\ref{alg:reach}, respectively.

\begin{theorem}
\label{thm:time}
Let $\depth$ denote the expected number of steps before reaching a BSCC and $\bsccsize$ the maximum size of a BSCC.
Further, let $\bscctime:=\max_{C\in\bscc;s,s'\in C}\expected[\text{time to}$ $\text{reach $s'$ from $s$}]$.
In particular, $\bscctime\in\mathcal O(\bsccsize/\pmin^\bsccsize)$.
Then the expected running time of Algorithms~\ref{alg:reach} and~\ref{alg:ltl} is at most 
$$\mathcal O(\simnum\cdot k_{\depth+\bsccsize}\cdot \bsccsize \cdot \bscctime)\,.$$
\end{theorem}
\begin{proof}
We show that the expected running time of each simulation is at most $k_{\depth+\bsccsize}\cdot \bsccsize \cdot \bscctime$.
Since the expected number of states visited is bounded by $\depth+\bsccsize$, the expected number of candidates on a run is less than $2(\depth+\bsccsize)-1$. 
Since $k_i$ grows linearly in $i$ it is sufficient to prove that the expected time to visit each state of a BSCC once (when starting in BSCC) is at most $\bsccsize \cdot \bscctime$. 
We order the states of a BSCC as $s_1,\ldots,s_b$, then the time is at most $\sum_{i=1}^b \bscctime$, where $b\leq \bsccsize$. This yields the result since $\depth\in\mathcal O(k_{\depth+\bsccsize}\cdot \bsccsize \cdot \bscctime)$.

It remains to prove that $\bscctime\leq \bsccsize/\pmin^\bsccsize$. Let $s$ be a state of a BSCC of size at most $\bsccsize$. 
Then, for any state $s'$ from the same BSCC, the shortest path from $s$ to $s'$ has length at most $\bsccsize$ and probability at least $\pmin^\bsccsize$. 
Consequently, if starting at $s$, we haven't reached $s'$ after $B$ steps with probability at most $1-\pmin^{\bsccsize}$, and we are instead in some state $s''\neq s'$, from which, again, the probability to reach $s'$ within $B$ steps at least $\pmin^{\bsccsize}$. 
Hence, the expected time to reach $s'$ from $s$ is at most 
$$
\sum_{i=1}^\infty \bsccsize\cdot i(1-\pmin^\bsccsize)^{i-1}\pmin^\bsccsize, 
$$
where $i$ indicates the number of times a sequence of $B$ steps is observed.
The series can be summed by differentiating a geometric series.
As a result, we obtain a bound 
${\bsccsize}/{p^{\bsccsize}}$.
\QED\end{proof}
Systems that have large deep BSCCs require longer time to reach for the required level of confidence.
However, such systems  are often difficult to handle also for other methods agnostic of the topology. For instance, correctness of \cite{sbmf11} on the example in Figure \ref{fig:self-looping states} relies on 
the termination probability $\pt$ being at most $1-\lambda$, which is less than $2^{-n}$ here. Larger values lead to incorrect results and smaller values to paths of exponential length.
Nevertheless, our procedure usually runs faster than the bound suggest; for detailed discussion see Appendix \ref{app:match}.



%
%
\smallskip

\textbf{Conclusion.}
To the best of our knowledge, we propose the first statistical model-checking 
method that exploits the information provided by each simulation run 
on the fly, in order to detect quickly a potential BSCC, and verify LTL 
properties with the desired confidence. 
This is also the first application of SMC to quantitative properties such as 
mean payoff. 
We note that for our method to work correctly, the precise value of $\pmin$ is not necessary, but a lower bound is sufficient. This lower bound can come from domain knowledge, or can be inferred directly from description of white-box systems, such as the \textsc{Prism} benchmark.

The approach we present is not meant to replace the other methods, but rather to be an addition to the repertoire of available approaches. Our method is particularly valuable for models that have small BSCCs and huge state space, such as many of the \textsc{Prism} benchmarks.

In future work, we plan to investigate the applicability of our method to Markov decision 
processes, and to deciding language equivalence between two Markov chains.


\bibliographystyle{plain}
\bibliography{main}

\newpage
\appendix

\section*{Appendix}

\section{Detailed experiments}
\label{app:de}
Table \ref{tab:det_results} shows detailed experimental result for unbounded reachability. 
Compared to Table \ref{tab:results} we included: 1) experiments for the $\eterm$ method with two other values of $\pt$, 2) we report the number  of sampled paths as ``samples,''
and 3) we report the average length of sampled paths as ``path length.''
Topology-agnostic methods, such as $\ecand$ and $\eterm$, cannot be compared directly with topology-aware methods, such as $\eanal$ and $\emc$, however for reader's curiosity
we highlighted in the table the best results among \emph{all} methods.

We observed that in the ``herman'' example the SMC algorithms work unusually slow. 
This problem seems to be caused by a bug in the original sampling engine of \textsc{Prism} and it appears that all SMC algorithms suffer equally from this problem.

\begin{landscape}
\begin{table}[ht]
\caption{Detailed experimental results for unbounded reachability. Simulation parameters: $\alpha=\beta=\varepsilon=0.01$, $\delta=0.001$. 
$\tout$ means a timeout or memory out, and $\wr$ means that the reported result was incorrect.
 Our approach is denoted by $\ecand$ here.
Highlights show the best result among \emph{all} methods. }
\label{tab:det_results}
\centering
\resizebox{\columnwidth}{!}{
\begin{tabular}{c|ccc|ccc|ccc|ccc||cccc|c} 
\hline
Example  & \multicolumn{3}{c|}{$\ecand$} & \multicolumn{3}{c|}{$\eterm, \pt=10^{-3}$} & \multicolumn{3}{c|}{$\eterm, \pt=10^{-4}$} & \multicolumn{3}{c||}{$\eterm, \pt=10^{-5}$} & \multicolumn{4}{c|}{$\eanal$} & $\emc$\\
name & 	time & samples & path length & time & samples & path length & time & samples & path length & time & samples & path length & time & samples & path length & analysis & time\\
\hline
bluetooth(4)                              & \ba{  2.6s} &   243     &   499     &   185.0s    & 43764     &   387     &    16.4s    &  3389     &   484     &     6.4s    &   463     &   495     &    83.2s    &   219     &   502     &  80.4s    &    78.2s      \\
bluetooth(7)                              & \ba{  3.8s} &   243     &   946     &   697.4s    & 106732     &   604     &    50.2s    &  6480     &   897     &    10.2s    &   792     &   931     &   284.4s    &   219     &   937     & 281.1s    &   261.2s      \\
bluetooth(10)                             & \ba{  5.0s} &   243     &  1391     & $\tout$     &   -       &   -       &   109.2s    &  9827     &  1292     &    15.0s    &   932     &  1380     & $\tout$     &   -       &   -       &   -       & $\tout$       \\
\hline
brp(500,500)                              &     7.6s    &   230     &  3999     & \ba{  3.2s} &   258     &   963     &    13.8s    &   258     &  9758     &   107.2s    &   258     & 104033     &    35.6s    &   207     &  3045     &  30.7s    &   103.0s      \\
brp(2K,2K)                                &    20.4s    &   230     & 13000     & \ba{  3.4s} &   258     &  1029     &    17.2s    &   258     &  9127     &   115.0s    &   258     & 98820     &   824.4s    &   207     & 12167     & 789.9s    & $\tout$       \\
brp(10K,10K)                              &    89.2s    &   230     & 61999     & \ba{  3.6s} &   258     &   960     &    15.8s    &   258     & 10059     &   109.4s    &   258     & 96425     & $\tout$     &   -       &   -       &   -       & $\tout$       \\
\hline
crowds(6,15)                              &   3.6s &   395     &   879     &    29.2s    &  7947     &   878     &   253.2s    &  7477     &  8735     & $\tout$     &   -       &   -       & \ba{  2.0s} &   400     &    85     &   0.7s    &    19.4s      \\
crowds(7,20)                              &   4.0s &   485     &   859     &    32.6s    &  9378     &   850     &   283.8s    &  8993     &  8464     & $\tout$     &   -       &   -       & \ba{  2.6s} &   473     &    98     &   1.1s    &   347.8s      \\
crowds(8,20)                              &   5.6s &   830     &   824     &    38.2s    & 11405     &   821     &   340.0s    & 10574     &  8132     & $\tout$     &   -       &   -       & \ba{  4.0s} &   793     &   110     &   1.9s    & $\tout$       \\
\hline
eql(15,10)                                & \ba{ 16.2s} &  1149     &   652     &   303.2s    & 28259     &   628     & $\tout$     &   -       &   -       & $\tout$     &   -       &   -       &   151.8s    &  1100     &   201     & 145.1s    &   110.4s      \\
eql(20,15)                                & \ba{ 28.8s} &  1090     &  1299     &   612.8s    & 44048     &   723     & $\tout$     &   -       &   -       & $\tout$     &   -       &   -       &   762.6s    &   999     &   347     & 745.4s    &   606.6s      \\
eql(20,20)                                & \ba{ 31.4s} &  1071     &  1401     & $\tout$     & 11408     &   156     & $\tout$     &   -       &   -       & $\tout$     &   -       &   -       & $\tout$     &   -       &   -       &   -       & $\tout$       \\
\hline
herman(17)                                &  23.0s &   243     &    30     &   257.6s    &  2101     &    30     &    33.6s    &   381     &    32     &    29.0s    &   279     &    31     &    21.6s    &   219     &    30     &   0.1s    & \ba{ 1.2s}   \\
herman(19)                                & 96.8s &   243     &    40     & $\tout$     &   -       &   -       &   134.0s    &   355     &    38     &   254.4s    &   279     &    40     &    86.2s    &   219     &    38     &   0.1s    & \ba{  1.2s}   \\
herman(21)                                & 570.0s &   243     &    46     & $\memout$   &   -       &   -       & $\tout$     &   -       &   -       & $\memout$   &   -       &   -       &   505.2s    &   219     &    48     &   0.1s    & \ba{  1.4s}   \\
\hline
leader(6,6)                               & \ba{  5.0s} &   243     &     7     &     7.6s    &   437     &     7     &     5.4s    &   258     &     7     &     \ba{5.0s}    &   258     &     7     &   536.6s    &   219     &     7     & 530.3s    &   491.4s      \\
leader(6,8)                               & \ba{ 23.0s} &   243     &     7     &    62.4s    &   560     &     7     &    26.0s    &   279     &     7     &    26.2s    &   258     &     7     & $\memout$   &   -       &   -       &   -       & $\memout$     \\
leader(6,11)                              & \ba{153.0s} &   243     &     7     & $\tout$     &   -       &   -       &   174.8s    &   279     &     7     &   776.8s    &   258     &     7     & $\memout$   &   -       &   -       &   -       & $\memout$     \\
\hline
nand(50,3)                                & \ba{  7.0s} &   899     &  1627     &   570.6s    & 140880     &   846     &   231.2s    & 21829     &  4632     & $\tout$     &   -       &   -       &    36.2s    &  1002     &  1400     &  31.0s    &   272.0s      \\
nand(60,4)                                & \ba{  6.0s} &   522     &  2431     & $\tout$     &   -       &   -       &   275.2s    & 25250     &  4494     & $\tout$     &   -       &   -       &    60.2s    &   458     &  2160     &  56.3s    & $\tout$       \\
nand(70,5)                                & \ba{  6.8s} &   391     &  3343     & $\tout$     &   -       &   -       &   370.2s    & 30522     &  4643     & $\tout$     &   -       &   -       &   148.2s    &   308     &  3080     & 144.2s    & $\tout$       \\
\hline
tandem(500)                               &     2.4s    &   243     &   501     &    59.6s    & 43156     &   394     &     6.4s    &  3318     &   489     & \ba{  2.0s} &   412     &   500     &     4.6s    &   219     &   501     &   3.0s    &     3.4s      \\
tandem(1K)                                & \ba{  2.6s} &   243     &  1001     &   328.4s    & 114288     &   632     &    19.2s    &  6932     &   954     &     3.4s    &   858     &   995     &    17.0s    &   219     &  1001     &  12.7s    &    13.0s      \\
tandem(2K)                                & \ba{  3.4s} &   243     &  2001     & $\tout$     &   -       &   -       &    72.4s    & 14881     &  1811     &     6.6s    &  1093     &  1985     &    62.4s    &   219     &  2001     &  59.8s    &    59.4s      \\
\hline
gridworld(300)                            & \ba{  8.2s} &  1187     &   453     &   214.4s    & 46214     &   349     &    81.6s    & 18678     &   437     &    77.4s    & 16663     &   449     & $\memout$   &   -       &   -       &   -       & $\memout$     \\
gridworld(400)                            & \ba{  8.4s} &  1047     &   543     &   274.8s    & 53152     &   399     &   100.6s    & 18909     &   531     &    93.0s    & 16674     &   548     & $\memout$   &   -       &   -       &   -       & $\memout$     \\
gridworld(500)                            & \ba{  5.8s} &   480     &   637     &   277.4s    & 57263     &   431     &   109.4s    & 18025     &   605     &   104.4s    & 15684     &   627     & $\memout$   &   -       &   -       &   -       & $\memout$     \\
\hline
Fig.\ref{fig:self-looping states}(16)                    &    58.6s    &   128     & 140664     & $\tout$     &   -       &   -       & $\tout$     &   -       &   -       & $\tout$     &   -       &   -       &    23.4s    &   115     & 141167     &   0.4s    &     \ba{2.0s}      \\
Fig.\ref{fig:self-looping states}(18)                    & $\tout$     &   -     & -     &     2.8s    &   258     &  1015     & $\tout$     &   -       &   -       & $\tout$     &   -       &   -       &    74.8s    &   115     & 537062     &   1.8s    &     \ba{2.0s}      \\
Fig.\ref{fig:self-looping states}(20)                    & $\tout$     &   -       &   -       &     $\wr$    &   -     &  -     & $\tout$     &   -     &  -     & $\tout$     &   -       &   -       &   513.6s    &   119     & 2195265     &   11.3s    &     \ba{2.0s}      \\
\hline
Fig.\ref{fig:add}(1K,5)                   &   7.8s &  1109     &  2489     & $\tout$     &   -       &   -       &   218.2s    & 23968     &  5916     & $\tout$     &   -       &   -       &     3.2s    &   896     &  1027     &   0.5s    & \ba{  1.2s}   \\
Fig.\ref{fig:add}(1K,50)                  &  12.4s &  1115     &  4306     & $\tout$     &   -       &   -       &   211.8s    & 23908     &  5880     & $\tout$     &   -       &   -       &     3.6s    &   881     &  1037     &   0.7s    & \ba{  1.0s}   \\
Fig.\ref{fig:add}(1K,500)                 &   431.0s    &  1002     & 177777     & $\tout$     &   -       &   -       & 218.6s & 23951     &  5903     & $\tout$     &   -       &   -       &     3.6s    &   968     &  1042     &   1.0s    & \ba{  1.2s}   \\
Fig.\ref{fig:add}(10K,5)                  &    52.2s    &  1161     & 20404     & \ba{  2.6s} &   258     &  1072     & $\tout$     &   -       &   -       & $\tout$     &   -       &   -       &    42.2s    &  1057     & 10100     &  25.4s    &    25.6s      \\
Fig.\ref{fig:add}(100K,5)                 &   604.2s    &  1331     & 200399     & \ba{  2.6s} &   258     &   981     &     5.4s    &   258     &  9939     & $\tout$     &   -       &   -       & $\tout$     &   -       &   -       &   -       & $\tout$       \\
\hline

\end{tabular}
}
\end{table}
\end{landscape}

\section{Implementation details}
\label{app:id}
In our algorithms we frequently check whether the simulated path
contains a candidate with the required strength.  To reduce the
time needed for this operation we use two optimization: 1) we record
SCs visited on the path, 2) we check if a candidate has been found every
$C_b\geq 1$ steps.  Our data structure records the sequence of SCs that
have been encountered on the simulated path. The candidate of the path
is then the last SC in the sequence. We also record the number of
times each state in the candidate has been encountered. By using this
data structure we avoid traversing the entire path every time we
check if a strong $k$-candidate has been reached.

To further reduce the overhead, we update our data structure every $C_b$
steps (in our experiments $C_b=1000$). Figure \ref{fig:steps_extreme}
shows the impact of $C_b$ on the running time for two Markov
chains.  The optimal value of $C_b$ varies among examples, however
experience shows that $C_b \approx 1000$ is a reasonable choice.

\begin{figure}[h]
  \centering
\input{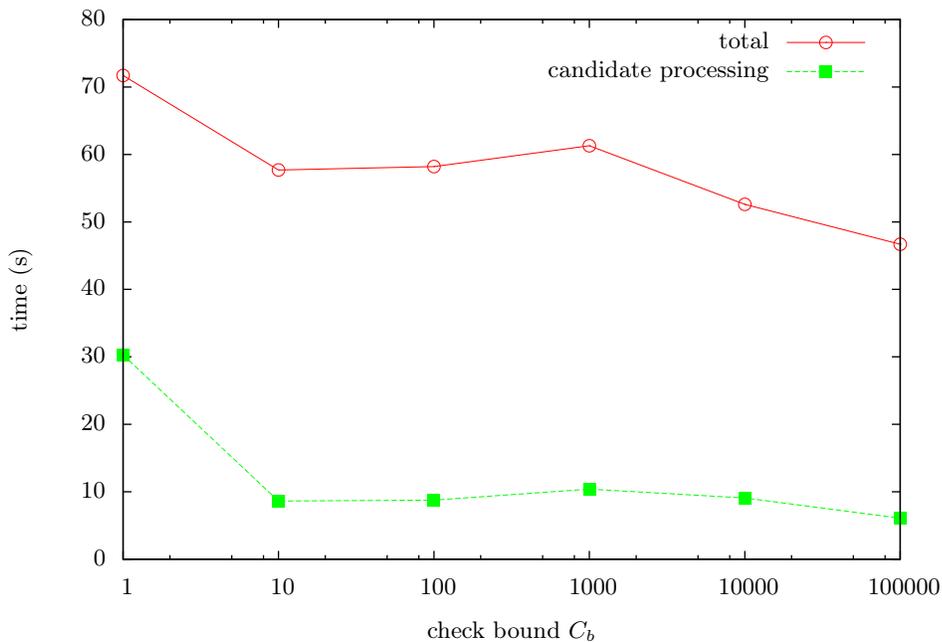}
  \caption{Total running time and time for processing candidates for a Markov chain in Figure \ref{fig:self-looping states} depending on the check bound $C_b$.}
  \label{fig:steps_extreme}
\end{figure}

\begin{figure}[h]
  \centering
\input{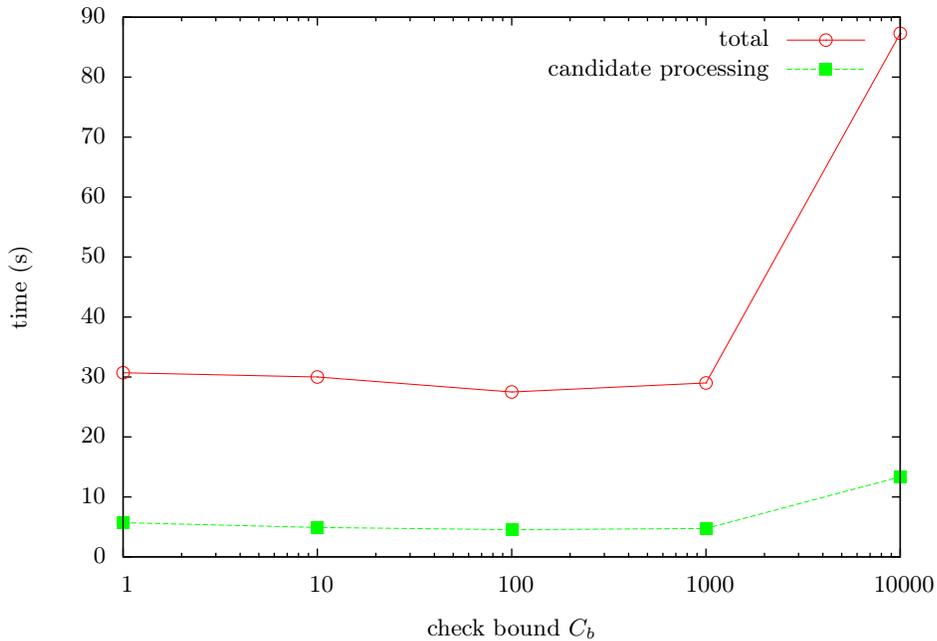}
  \caption{Total running time and time for processing candidates for the 'eql(20,20)' benchmark depending on the check bound $C_b$.}
  \label{fig:steps_eql}
\end{figure}

\section{Theoretical vs.\ empirical running time}
\label{app:match}
In this section, we compare  the theoretical upper bound on running time given in 
Theorem \ref{thm:time} to  empirical data.
We omit the number of simulation runs (term $\simnum$ in the theorem), and report only the logarithm of average simulation length.
Figures \ref{fig:extreme_pred}, \ref{fig:scalen_pred} and \ref{fig:scalem_pred} present
the comparison for different topologies of Markov chains.
In Figure \ref{fig:extreme_pred} we present the comparison for the worst-case Markov chain, which requires the longest paths
to discover the BSCCs as a $k$-candidate.
This Markov chain is like the one in Figure \ref{fig:self-looping states}, but where the last state has a single outgoing transition to the initial state.
Figure \ref{fig:scalen_pred} suggests that the theoretical bound  can be a good predictor of running time with respect to the depth of the system,
 however, Figure \ref{fig:scalem_pred} shows that it is very conservative with respect to the size of BSCCs.

\begin{figure}[h]
  \centering
\input{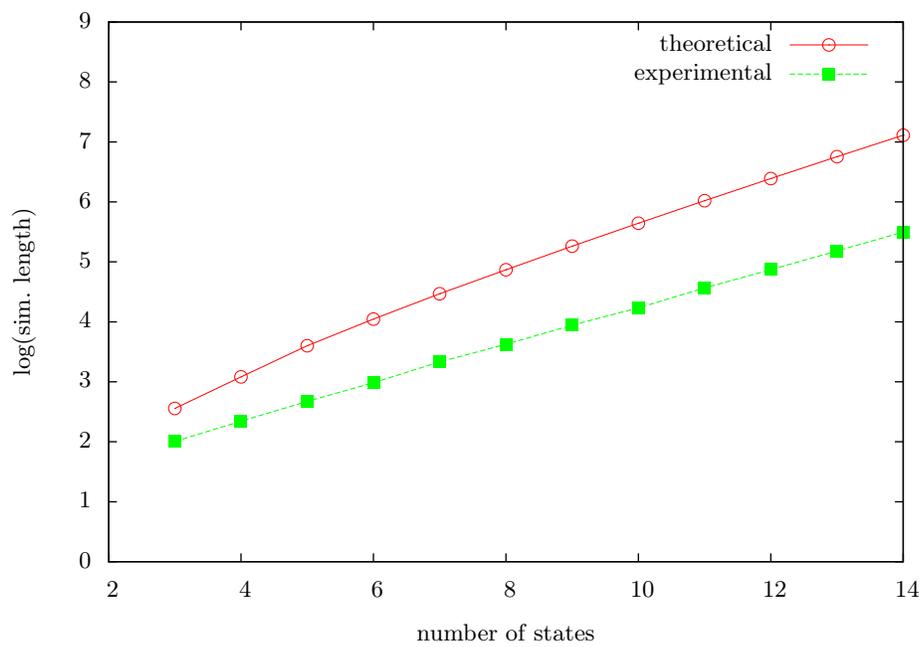}
  \caption{Average length of simulations for a Markov chain like in Figure \ref{fig:self-looping states},
    but where the last state has a single outgoing transition to the initial state.}
  \label{fig:extreme_pred}
\end{figure}

\begin{figure}[h]
  \centering
\input{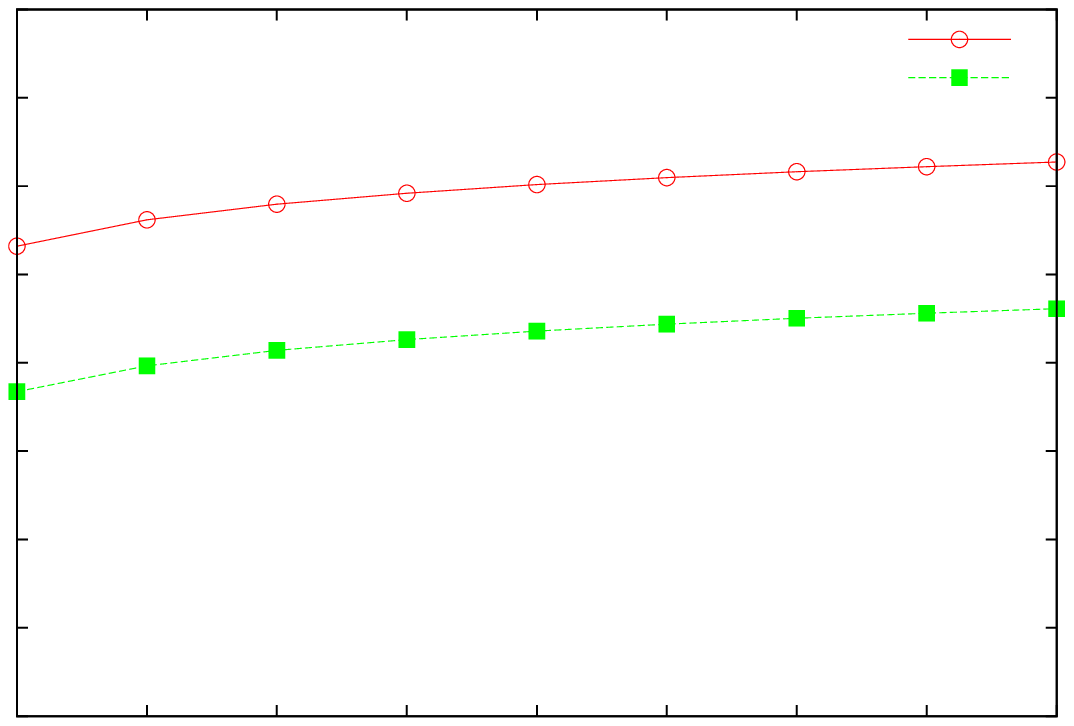}
  \caption{Average length of simulations for the MC in Figure \ref{fig:add}, where $M=5$ and $N$ varies.}
  \label{fig:scalen_pred}
\end{figure}

\begin{figure}[h]
  \centering
\input{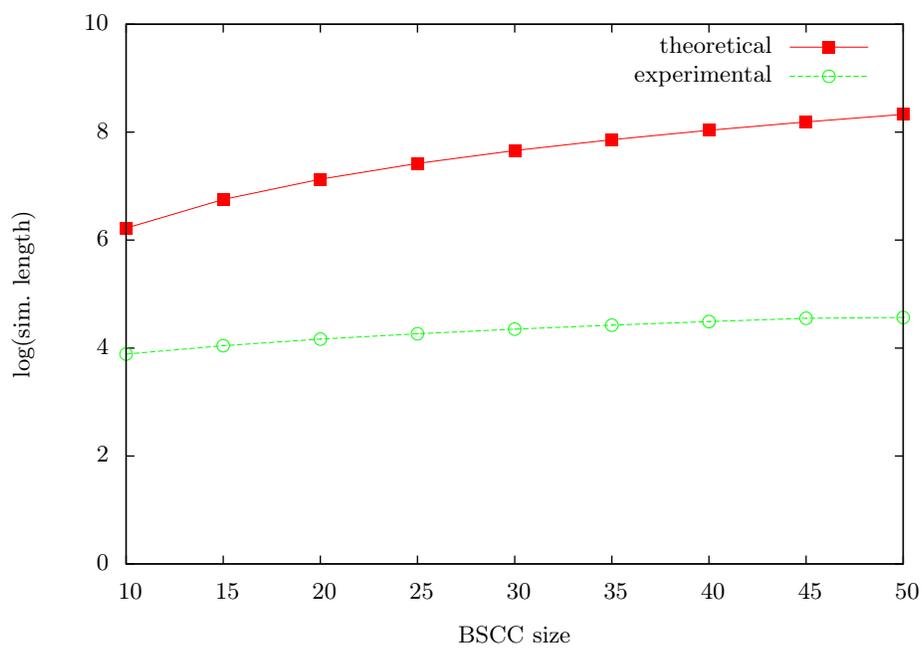}
  \caption{Average length of simulations for the MC in Figure \ref{fig:add}, where $N=1000$ and $M$ varies.}
  \label{fig:scalem_pred}
\end{figure}

\clearpage


\end{document}